\documentclass[11pt]{article} 
\usepackage{rpmacros}
\usepackage[dvipsnames,svgnames,x11names,hyperref]{xcolor}
\usepackage{amsthm}
\usepackage{enumitem}
\usepackage{xfrac}
\usepackage[T1]{fontenc}
\usepackage[noline, ruled, linesnumbered]{algorithm2e}
\let\oldnl\nl
\newcommand{\nonl}{\renewcommand{\nl}{\let\nl\oldnl}}
\SetCommentSty{textnormal}
\SetKwComment{Comment}{/* }{ */}
\DontPrintSemicolon

\usepackage{stmaryrd}
\usepackage{gitinfo}
\usepackage{enumitem}
\usepackage{fullpage}
\usepackage{nicefrac}
\usepackage[colorlinks,pagebackref]{hyperref}
\hypersetup{
  linkcolor=[rgb]{0,0,0.4},
  citecolor=[rgb]{0, 0.4, 0},
  urlcolor=[rgb]{0.6, 0, 0}
}
\usepackage{mathrsfs,bbm}
\usepackage{todonotes}
\usepackage{tikz}
\usetikzlibrary{decorations.pathreplacing,backgrounds}
\usepackage{multirow}
\usepackage{setspace}
\usepackage{amsthm}
\usepackage{thmtools,thm-restate}
\usepackage[nameinlink,capitalise]{cleveref}

\counterwithin{equation}{section}

\declaretheoremstyle[bodyfont=\it,qed=\qedsymbol]{noproofstyle}

\declaretheorem[numberlike=equation]{observation}

\declaretheorem[name=Observation,numbered=no]{observation*}

\declaretheorem[numberlike=equation]{theorem}

\declaretheorem[name=Theorem,numbered=no]{theorem*}

\declaretheorem[numberlike=equation]{lemma}
\declaretheorem[name=Lemma,numbered=no]{lemma*}

\declaretheorem[numberlike=equation]{corollary}
\declaretheorem[name=Corollary,numbered=no]{corollary*}

\declaretheorem[name=Proposition,numbered=no]{proposition*}

\declaretheorem[numberlike=equation,name=Claim]{claim}
\declaretheorem[name=Claim,numbered=no]{claim*}

\declaretheorem[name=Conjecture,numbered=no]{conjecture*}

\declaretheorem[name=Question,numbered=no]{question*}

\declaretheoremstyle[bodyfont=\it,qed=$\lozenge$]{defstyle} 

\declaretheorem[numberlike=equation,style=defstyle]{definition}
\declaretheorem[unnumbered,name=Definition,style=defstyle]{definition*}

\declaretheorem[unnumbered,name=Example,style=defstyle]{example*}

\declaretheorem[unnumbered,name=Notation=defstyle]{notation*}

\declaretheorem[unnumbered,name=Construction,style=defstyle]{construction*}

\declaretheorem[unnumbered,name=Remark,style=defstyle]{remark*}

\newcommand{\Res}{\operatorname{Res}}

\newcommand{\thref}[1]{\hyperref[#1]{Theorem \ref{#1}}}
\newcommand{\lemref}[1]{\hyperref[#1]{Lemma \ref{#1}}} 
\usepackage{nth}
\usepackage{intcalc}
\usepackage{etoolbox}
\usepackage{xstring}

\usepackage{ifpdf}
\ifpdf
\else
\usepackage[quadpoints=false]{hypdvips}
\fi

\newcommand{\ECCCurl}[2]{http://eccc.hpi-web.de/report/\ifnumcomp{#1}{>}{93}{19}{20}#1/#2/}

\newcommand{\shortECCC}[2]{\texttt{\href{http://eccc.hpi-web.de/report/\ifnumcomp{#1}{>}{93}{19}{20}#1/#2/}{eccc:TR#1-#2}}}

\newcommand{\parseECCC}[1]{
\StrSubstitute{#1}{TR}{}[\tmpstring]%
\IfSubStr{\tmpstring}{/}{ 
\StrBefore{\tmpstring}{/}[\ecccyear]%
\StrBehind{\tmpstring}{/}[\ecccreport]%
}{
\StrBefore{\tmpstring}{-}[\ecccyear]%
\StrBehind{\tmpstring}{-}[\ecccreport]%
}%
\shortECCC{\ecccyear}{\ecccreport}}

\newcommand*\samethanks[1][\value{footnote}]{\footnotemark[#1]}

\newif\ifanon
\anonfalse

\newif\ifdraft
\drafttrue
\ifdraft
\newcommand{\phnote}[1]{\todo[color=red!100!green!33,size=\footnotesize]{ph: #1}}
\newcommand{\sriprahladhuvacha}[1]{\todo[color=red!100!green!33,inline,size=\small]{ph: #1}}
\newcommand{\MKnote}[1]{\todo[color=orange,size=\footnotesize]{mk: #1}}
\newcommand{\mrinalsays}[1]{\todo[color=orange,inline,size=\small]{mk: #1}}
\newcommand{\SCnote}[1]{\todo[color=green!20!blue!40!white,size=\footnotesize]{sc: #1}}
\newcommand{\sohamsays}[1]{\todo[color=green!20!blue!40!white,inline,size=\small]{sc: #1}}

\else
\newcommand{\phnote}[1]{}
\newcommand{\sriprahladhuvacha}[1]{}
\newcommand{\SCnote}[1]{}
\newcommand{\MKnote}[1]{}
\newcommand{\mrinalsays}[1]{}
\newcommand{\sohamsays}[1]{}

\fi

\newcommand{\ignore}[1]{}

\renewcommand{\epsilon}{\varepsilon}
\renewcommand{\hat}{\widehat}

\newcommand{\CombSplits}{\textsc{Combine-Splits}}
\newcommand{\Split}{\textsc{Split}}

\crefname{claim}{Claim}{Claims}

\title{Deterministic list decoding of Reed-Solomon codes}

\ifanon
\else
\author{Soham Chatterjee\thanks{Tata Institute of Fundamental Research, Mumbai, India. \href{mailto:soham.chatterjee@tifr.res.in}{\{soham.chatterjee}, \href{mailto:prahladh@tifr.res.in}{prahladh}, \href{mailto:mrinal@tifr.res.in}{mrinal\}@tifr.res.in}. Research supported by the Department of Atomic Energy, Government of India, under project number RTI4014, and partially by reserach grants from SERB, Google research and Premji Invest.}
\and
 {Prahladh Harsha\samethanks}
\and
 {Mrinal Kumar\samethanks}
 }
\fi

\date{}

\begin{document}

\maketitle
\begin{abstract}
  We show that Reed-Solomon codes of dimension $k$ and block length $n$  over any finite field $\F$ can be \emph{deterministically} list decoded from agreement $\sqrt{(k-1)n}$ in time $\poly(n, \log |\F|)$. 
  
  Prior to this work, the list decoding algorithms for Reed-Solomon codes, from the celebrated results of Sudan and Guruswami-Sudan, were either randomized with time complexity $\poly(n, \log |\F|)$ or were deterministic with time complexity depending polynomially on the characteristic of the underlying field. In particular, over a prime field $\F$, no deterministic algorithms running in time $\poly(n, \log |\F|)$ were known for this problem.

  Our main technical ingredient is a deterministic algorithm for solving the bivariate polynomial factorization instances that appear in the algorithm of Sudan and  Guruswami-Sudan with only a $\poly(\log |\F|)$ dependence on the field size in its time complexity for every finite field $\F$. While the question of obtaining efficient deterministic algorithms for polynomial factorization over finite fields is a fundamental open problem even for \emph{univariate} polynomials of degree $2$, we show that additional information from the received word can be used to obtain such an algorithm for instances that appear in the course of list decoding Reed-Solomon codes.
\end{abstract}

\tableofcontents

\section{Introduction}

Reed--Solomon (RS) codes are among the most important and well-studied families of error-correcting codes, notable for their algebraic simplicity and optimal trade-off between rate and distance. Their simple algebraic structure lends to extremely efficient encoding and decoding algorithms, and RS codes are widely used both in theory and practice. An \([n,k]\) RS code over a finite field \(\mathbb{F}\) encodes a message polynomial \(f(X)\) of degree strictly less than \(k\) by evaluating it at \(n\) distinct field elements \(\alpha_1, \dots, \alpha_n\), yielding the codeword
\[
(f(\alpha_1), f(\alpha_2), \dots, f(\alpha_n)) \in \mathbb{F}^n.
\]
These codes attain the Singleton bound and have distance exactly \(n-k+1\). 

\paragraph{Unique and list decoding.}
In the \emph{unique decoding problem}, given a received word, the goal is to find the unique codeword, if it exists, within a Hamming ball of radius half the minimum distance, which for RS codes equals \((n-k+1)/2\). Classical algorithms such as those of Berlekamp-Welch, Berlekamp-Massey and Peterson efficiently achieve this bound.
However, when the corruption exceeds this radius, multiple codewords may have significant agreement with the received word.  
The \emph{list decoding} problem, introduced by Elias \cite{Elias1957} and Wozencraft \cite{Wozencraft1958}, relaxes this requirement of uniqueness: instead of demanding uniqueness, the decoder outputs all codewords that are sufficiently close to the received word. 

A sequence of breakthroughs, beginning with Sudan's algorithm~\cite{Sudan1997} and refined by Guruswami and Sudan~\cite{GuruswamiS1999}, established that RS codes can be efficiently list decoded from agreement roughly \(\sqrt{(k-1)n}\), corresponding to the Johnson radius, which is significantly larger than the half-the-minimum distance. These algorithms hinge on two key algebraic steps: (1) constructing a low-degree bivariate polynomial that vanishes at all received points with multiplicity, and (2) \emph{factorizing} this polynomial to recover all potential message polynomials.
 
\paragraph{Derandomization and the factorization barrier.}
Both the Sudan and Guruswami--Sudan algorithms rely on the factorization of bivariate polynomials over finite fields.
The known algorithms for polynomial factorization include algorithms for univariate polynomials due to Berlekamp~\cite{Berlekamp1967-factoring,Berlekamp1970} and Cantor--Zassenhaus~\cite{CantorZ1981}, and algorithms for multivariates developed in the 1980s (e.g. Kaltofen~\cite{Kaltofen1985}, Lenstra~\cite{Lenstra1985}). These algorithms run in polynomial time\footnote{For the multivariate setting, the input and the output polynomials are generally succinctly represented by algebraic circuits. } and are highly efficient in practice, but are unfortunately \emph{randomized}. Despite decades of effort, obtaining efficient deterministic polynomial-time algorithms for polynomial factorization over finite fields\footnote{Over the field of rational numbers, deterministic polynomial time algorithms for univariate factorization were designed in the work of Lenstra, Lenstra \& Lovasz \cite{LenstraLL1982}. } remains a fundamental open problem in computational algebra.

Even in the univariate case, where one seeks to factor a degree-\(d\) polynomial \(f \in \mathbb{F}_q[X]\), all known deterministic algorithms either have running time polynomial in the \emph{characteristic} of the field~\cite{Berlekamp1967-factoring,Berlekamp1970} or apply only to restricted classes of polynomials or work under unproven assumptions.
For instance, while it is known that Shoup's improved deterministic factorization algorithm \cite{Shoup1990} is efficient (i.e., takes time \(\poly(d,|\log|\F|)\)) for \emph{most} polynomials, the best provable worst-case guarantee is \(\sqrt{p}\cdot \poly(d,\log |\F|)\). Several deterministic algorithms based on the Generalized Riemann Hypothesis (GRH) have also been developed; however, these approaches not only rely on an unproven assumption but, like Shoup's algorithm, fail to handle \emph{all} polynomials. In short, unconditional deterministic polynomial-time factorization over arbitrary finite fields remains elusive, even for quadratic polynomials.  

This limitation directly impacts derandomization efforts in list decoding; since both Sudan and Guruswami--Sudan algorithms use the factorization step. As in the general polynomial factorization setting, all known deterministic algorithms for the list-decoding RS codes run in time polynomial in \(n\) and in the characteristic of the underlying field. Consequently, when the characteristic is polynomial in \(n\) (as is the case in most coding theoretic settings), these algorithms do yield a fully deterministic list-decoding procedure. However, no truly deterministic algorithm with running time polynomial in \((n,\log |\F|)\) was known for all settings of \(n\) and the finite field \(\F\). In particular, even for prime fields \(\mathbb{F}_p\) with \(p \gg n\) (superpolynomial), the existence of such a deterministic algorithm remained open.

\paragraph{Our results.}
We give the first \emph{deterministic list-decoding algorithms for RS codes} up to the Johnson radius that run in time polynomial in both the block length and logarithm of the field size.  

More precisely, we prove the following theorem. 
\begin{theorem}\label{thm:intro-main}
  There is a deterministic algorithm that for every finite field $\F$ and parameters $n > k \in \N$ runs in  time $\poly(n, \log |\F|)$ and list decodes Reed-Solomon codes of block length $n$ and dimension $k$ over $\F$ from agreement greater than $\sqrt{(k-1)n}$.
\end{theorem}
It would be interesting if we can obtain near-linear time versions of these deterministic algorithms (a la Alekhnovich \cite{Alekhnovich2005}).

Our work demonstrates that even though general deterministic factorization remains elusive, the algebraic structure inherent in coding-theoretic instances can circumvent this technical issue. This result is an addition to the broader \emph{derandomization program} in computational complexity, offering yet another concrete setting where randomness can be provably removed without assuming any unproven hypotheses.

\paragraph{Techniques.}
Our main technical contribution is a new \emph{deterministic algorithm for bivariate polynomial factorization} in the structured setting that arises within the algebraic list-decoding framework.  
While the general problem of deterministic factorization remains open, we exploit additional structure available from the received word, specifically, knowledge of its evaluation pattern and multiplicities, to remove randomness in this restricted but highly relevant setting.  

This idea of exploiting additional information from the codeword also appears in the work of Kalev and Ta-Shma~\cite{KalevT2022} in their analysis of a multiplicity-code-based lossless condenser. Inspired by their approach, we can obtain the following derandomization of Sudan's algorithm.
\begin{theorem}\label{thm:intro-sudan}
There is a deterministic algorithm that, for every finite field~$\mathbb{F}$ and parameters $n,k \in \mathbb{N}$, runs in time $\mathrm{poly}(n, \log |\mathbb{F}|)$ and list decodes Reed--Solomon codes of block length~$n$ and dimension~$k$ over~$\mathbb{F}$ from agreement greater than $\sqrt{2(k-1)n}$.
\end{theorem}
Such a derandomization of Sudan's algorithm for Reed-Solomon codes (and their algebraic geometric variants due to Shokrollahi and Wasserman \cite{ShokrollahiW1999}) was discovered by Augot and Pecquet \cite{AugotP2000}\footnote{In an earlier version of this paper, we stated that such a derandomization was not previously known. We were unaware of the work of Augot and Pecquet \cite{AugotP2000} at that time. However, their approach is limited to the Sudan algorithm and does not extend to yield a derandomization of the Guruswami-Sudan algorithm,  which is the main focus of our work.}. Our derandomized variant of Sudan's algorithm employs Newton's iteration\footnote{Kalev-TaShma used a similar approach to solve linear differential equations in \cite{KalevT2022}.} to find a root $(Y-f(X))$ of a bivariate polynomial $Q(X,Y)$. Typical applications of Newton's iteration use a point $(\alpha,\beta)$ satisfying $Q(\alpha,\beta)=0$ and $\frac{\partial Q}{\partial Y}(\alpha,\beta)\neq 0$. This point is typically obtained by choosing an appropriate $\alpha$, and then using randomized univariate factorization algorithms to obtain a root $\beta$ of $Q(\alpha, Y)$. To do this factorization deterministically in the list-decoding setting, we note that at all the points $(\alpha_j,\beta_j)$ in the received word, $Q$ satisfies $Q(\alpha_j,\beta_j)=0$. If furthermore, one of these points satisfy $\frac{\partial Q}{\partial Y}(\alpha,\beta)\neq 0$, we can avoid the \emph{random} search for a suitable $(\alpha, \beta)$. Otherwise, if all points $(\alpha_j,\beta_j)$ of agreement of $f$ with the codeword satisfy $\frac{\partial Q}{\partial Y}(\alpha,\beta)= 0$, we can use this to show that $(Y-f(X))$ is a root of $\frac{\partial Q}{\partial Y}(X,Y)$, a lower-degree polynomial. Applying this idea recursively, yields the above deterministic variant of Sudan's algorithm.

However, this idea alone is insufficient to derandomize the Guruswami--Sudan algorithm: in both Sudan's list decoding and Kalev--Ta-Shma's condenser, the multiplicity used in the interpolation step is \(m=1\), whereas in the Guruswami--Sudan algorithm the multiplicity is considerably higher.

We follow an alternate approach using Hensel lifting instead of Newton's iteration. Standard applications of Hensel lifting, like Newton's iteration, are seeded using an initial random step which we cannot afford. To circumvent this, we do the following: for each codeword point $(\alpha_j,\beta_j)$, we first deterministically factorize the interpolation polynomial $Q(\alpha_j,Y)$ into coprime factors as follows: $Q(\alpha_j,Y) = (Y-\beta_j)^{m_j}\cdot \hat{P_j}(Y)$ for some $\hat{P_j}(Y)$ such that $\hat{P_j}(\beta_j)\neq 0$. We use this initial condition as a base case for the Hensel lifting and obtain a non-trivial local factorization of $Q$ corresponding to each  point $(\alpha_j,\beta_j)$ in the received word.  We then apply this local factorization recursively to the set of factors obtained, till the set stabilizes. We then show how to obtain the desired factors $(Y-f(X))$ corresponding to the polynomials $f$ in the list from the stabilized factorization. 

A detailed overview of these deterministic algorithms is presented in \cref{sec:overview}.

\section{Proof overview}\label{sec:overview}

In this section, we outline the main ideas behind the proof of \cref{thm:intro-main}. To set the stage, we first recall the structure of the list-decoding algorithms of Sudan~\cite{Sudan1997} and Guruswami--Sudan~\cite{GuruswamiS1999}. Let $w = \{(\alpha_j, \beta_j)\}_{j \in [n]}$ denote the received word, where $\alpha_1, \dots, \alpha_n \in \F$ are distinct evaluation points. The algorithms depend on parameters $m, D, t$, to be fixed later, and output a list $L$ of degree-$(k-1)$ univariate polynomials $f(X)$ that agree with $w$ on at least $t$ points. Both algorithms share the same two-step structure:

\begin{enumerate}
  \item \textbf{Interpolation.} 
  Find a nonzero polynomial $Q(X,Y)$ of $(1,k-1)$-degree\footnote{The $(a,b)$-degree of a monomial $x^iy^j$ is defined to be equal to $(ai + bj)$. The $(a,b)$-degree of a polynomial $Q$ is the maximum of the $(a,b)$-degrees of the non-zero monomials in the polynomial.} at most $D$ that vanishes at each $(\alpha_j,\beta_j)$ with multiplicity at least $m$. 
  (Formally, all Hasse derivatives of order less than $m$ vanish at each point.)
  \item \textbf{Factorization.} 
  Factor $Q(X,Y)$ over $\F$; for each factor of the form $(Y-f(X))$, output $f$ if $\deg f < k$ and $f$ agrees with $w$ on at least $t$ positions.
\end{enumerate}

The distinction between the two algorithms lies in the choice of parameters. 
Sudan's algorithm sets $m = 1$ and $D, t \approx \sqrt{2(k-1)n}$, 
while Guruswami--Sudan sets $m = \poly(n)$ 
to decode from agreement $t = \sqrt{(k-1)n}$, with $D \approx m\sqrt{(k-1)n}$. 
The use of higher multiplicities, one of the key innovations of~\cite{GuruswamiS1999}, has since become a standard tool in algebraic coding theory~\cite{Saraf2011}.

The interpolation step reduces to solving a homogeneous linear system in the coefficients of $Q$. The parameters are chosen so that the number of constraints is smaller than the number of variables and hence a nonzero solution always exists and can be found deterministically in $\poly(n, \log q)$ time by Gaussian elimination.

Randomness enters only in the \emph{factorization} step, which employs standard algorithms for bivariate polynomial factorization over $\F$. 
A natural approach to derandomizing the decoding algorithm would be to design a deterministic $\poly(d, \log |\F|)$-time algorithm for factoring arbitrary bivariate polynomials of degree $d$ over a finite field $\F$. 
However, such an algorithm is unknown even for univariate quadratics. 
The factorization instances that arise in the decoding algorithms of~\cite{Sudan1997,GuruswamiS1999} exhibit additional structure: the desired factors are always linear in $Y$ (of the form $(Y-f(X))$ for a univariate $f$ of degree less than $k$), and we have access to noisy evaluations of $f$. 
We show that this additional structure can be exploited to eliminate randomness efficiently.

The techniques differ slightly between the two parameter settings, and we first describe the simpler case corresponding to Sudan's algorithm before turning to the Guruswami--Sudan algorithm.

\subsection{Derandomizing Sudan's decoder}

As mentioned in the introduction, a derandomization of Sudan's algorithm was discovered earlier by Augot and Pecquet \cite{AugotP2000} using a Hensel lifting approach. We find it constructive to recall their result in our framework using Newton's iteration, as it provides a natural stepping stone toward the more general derandomization of the Guruswami-Sudan decoder.
In Sudan's algorithm, the interpolation step outputs a polynomial $Q(X,Y)$ that vanishes on each $(\alpha_j,\beta_j)$ and has $(1,k-1)$-degree less than $\sqrt{2(k-1)n}$. 
The analysis observes that for any polynomial $f(X)$, the univariate $Q(X,f(X))$ vanishes at every $\alpha_j$ where $f(\alpha_j)=\beta_j$. 
Hence, if the agreement between $f$ and the received word exceeds the degree of $Q(X,f(X))$ (bounded by the $(1,k-1)$-degree of $Q$), then $Q(X,f(X))$ must be identically zero or equivalently, $f(X)$ is a $Y$-root of $Q(X,Y)$. 
To recover such an $f$ from $Q$, we employ the classical Newton iteration technique, stated below.

\begin{restatable}[Newton Iteration~{\cite[Chapter 9]{GathenG-MCA}}]{lemma}{NI}\label{lem:NI}
  Let $\mathcal{R}=\F[X]$ be a polynomial ring and let $Q(X,Y)\in\mathcal{R}[Y]$.
  Suppose $\varphi\in\F\llbracket X\rrbracket$ satisfies $Q(X,\varphi)\equiv 0\bmod{\langle X\rangle^m}$ and $Q^{(0,1)}(0,\varphi(0))\neq 0$. 
  Then 
  \[
    \varphi':= \varphi- \frac{Q(X,\varphi)}{ Q^{(0,1)}(0,\varphi(0))}
  \]
  satisfies $Q(X,\varphi')\equiv 0\bmod{\langle X\rangle^{m+1}}$ and $\varphi'\equiv \varphi\bmod{\langle X\rangle^m}$. 
  Moreover, this extension is unique: any $\varphi''$ satisfying 
  $Q(X,\varphi'')\equiv 0\bmod{\langle X\rangle^{m+1}}$ and $\varphi''\equiv \varphi\bmod{\langle X\rangle^m}$ must also satisfy $\varphi''\equiv \varphi'\bmod{\langle X\rangle^m}$.
\end{restatable}

In our setting, the lemma implies that if we can find $(\alpha,\beta)\in\F^2$ such that $f(\alpha)=\beta$ and $Q^{(0,1)}(\alpha,\beta)\neq 0$, then starting from 
\[
  Q(X,\beta)\equiv 0\bmod{\langle X-\alpha\rangle},
\]
we can iteratively lift modulo powers of $\langle X-\alpha\rangle$ using \cref{lem:NI} to recover $f$. 
This process is fully deterministic and can be implemented in $\poly(\deg Q,\log|\F|)$ time over any finite field. 
Thus, recovering $f$ reduces to finding a point $(\alpha,\beta)$ such that $f(\alpha)=\beta$ and $Q^{(0,1)}(\alpha,\beta)\neq 0$.

Given the received word $w$, there are many candidate points $(\alpha_j,\beta_j)$ where $f(\alpha_j)=\beta_j$. 
Although we do not know which coordinates correspond to actual agreements, we can simply try all $n$ possibilities. 
If some $j$ satisfies $Q^{(0,1)}(\alpha_j,\beta_j)\neq 0$, then we can deterministically recover $f$ via \cref{lem:NI}.

The remaining case is for those close enough polynomials $f$ such that for every point of agreement $(\alpha_j,\beta_j)$, the derivative $Q^{(0,1)}(\alpha_j,\beta_j)=0$.  In this situation, $f$ is in fact a $Y$-root of $Q^{(0,1)}(X,Y)$.  The $Y$-degree of $Q^{(0,1)}$ is strictly smaller than that of $Q$, while its $(1,k-1)$-degree does not increase in the process.  Define $R_1(X)=Q^{(0,1)}(X,f(X))$. By assumption, $R_1(\alpha_j)=0$ at all points of agreement of $f$, implying $R_1(X)\equiv 0$.  Hence, $f(X)$ is also a $Y$-root of $Q^{(0,1)}(X,Y)$.

We now repeat the same reasoning with $Q^{(0,1)}$ in place of $Q$. 
If there exists $(\alpha_j,\beta_j)$ satisfying the conditions of \cref{lem:NI}, we recover $f$; 
otherwise, $f$ must also be a $Y$-root of $Q^{(0,2)}(X,Y)$, and so on. 
Consequently, the decoder simply tries all partial derivatives $Q^{(0,i)}(X,Y)$ for successive $i$ and all candidate points $(\alpha_j,\beta_j)$. 
For each pair, it invokes the Newton iteration subroutine. 
Correctness follows from the observation that for every codeword $f$ close enough to $w$, at least one such invocation succeeds. 
The technical details are deferred to \cref{sec:sudan}.

\subsection{Derandomizing the Guruswami--Sudan decoder}\label{sec:po-gs}

Given the previous section, it is natural to ask whether iterating over the $Y$-derivatives of $Q$ also works in the Guruswami--Sudan parameter regime. This direct approach fails. To decode from agreement $\sqrt{(k-1)n}$, the interpolation multiplicity is set to $m = \poly(n)$ 
, so $Q$ vanishes with multiplicity at least $m$ at each $(\alpha_j,\beta_j)$. Consequently, for any $f(X)$ agreeing with $w$ at $(\alpha_j,\beta_j)$, the univariate $Q(X,f(X))$ vanishes with multiplicity at least $m$ at $X=\alpha_j$, and hence we can conclude that $Q(X,f(X))$ is identically zero if
\begin{equation}\label{eqn:intro-1}
m\cdot \#\mathrm{agreements}(f,w) > (1,k-1)\text{-degree of } Q.
\end{equation}
One might try to recover $f$ by Newton iteration from an agreement point. However, since $m>1$, we have $Q^{(0,1)}(\alpha_j,\beta_j)=0$ for every $(\alpha_j,\beta_j)$; thus we can attempt to recurse on $Q^{(0,1)}$. For this to work, we would need
\begin{equation}\label{eqn:intro-2}
(m-1)\cdot \#\mathrm{agreements}(f,w) > (1,k-1)\text{-degree of }Q^{(0,1)}.
\end{equation}
But $(1,k-1)$-degree of $Q^{(0,1)}$ can be as large as $D-k$ where $D=(1,k-1)$-degree of $Q$, so relative to \eqref{eqn:intro-1} the left-hand side drops by at least $\sqrt{(k-1)n}$ while the right-hand side drops by only $k$. Thus \eqref{eqn:intro-2} need not hold, and the simple derivative recursion from the Sudan case does not extend.

To overcome this, we refine $Q$ via a local factor-type decomposition at each received point. Let $P=\prod_{i=1}^s P_i$ be the factorization of a \emph{monic} bivariate $P(X,Y)\in\F[X,Y]$ into irreducibles (with multiplicity). For simplicity, we will assume in this proof overview section that all the bivariates are monic in the variable $Y$. While the monic setting conveys most of the high level ideas of the proof, for our final proofs, we have to work with  non-monic bivariates and this leads to some additional technical difficulties that we handle in \autoref{sec:GS} and \autoref{sec:split}.  For $(\alpha,\beta)\in\F^2$, we partition the set $[s]$ corresponding to the factors of $P$, into 3 sets as follows:
\begin{align*}
  A(\alpha,\beta) &:= \{i \in [s] \mid P_i(\alpha,Y)=\gamma\,(Y-\beta)^{m_i} \text{ for some } m_i\ge 1,\ \gamma\in\F \setminus \{0\}\},\\
  B(\alpha,\beta) &:= \{i \in [s]\mid  P_i(\alpha,\beta)\neq 0, \deg_Y P_i\ge 1\},\\
  C(\alpha,\beta) &:= \{i \in [s] \mid P_i(\alpha,Y)=(Y-\beta)^{m_i}\cdot \hat P_i(Y) \text{ for some } m_i\ge 1, \deg_Y \widehat P_i\ge 1, \widehat P_i(\beta)\neq 0\}.
\end{align*}

Let $P_{A(\alpha,\beta)}:=\prod_{i\in A(\alpha,\beta)}P_i$ and $P_{B(\alpha,\beta)}$, $P_{C(\alpha,\beta)}$ be defined similarly. For the received word $w=(\alpha_j,\beta_j)_{j\in[n]}$, write $A_j,B_j,C_j$ and $P_{A,j},P_{B,j},P_{C,j}$ instead of $A(\alpha_j,\beta_j), B(\alpha_j,\beta_j), C(\alpha_j,\beta_j)$ and $P_{A(\alpha_j,\beta_j)}, P_{B(\alpha_j,\beta_j)}, P_{C(\alpha_j,\beta_j)}$ for $(\alpha_j,\beta_j)$.  

Since $P$ is monic in $Y$, it follows from the Gauss' lemma that each of its irreducible factors can be assumed to be monic in $Y$. And, thus, $P$ must equal the product of $P_{A,j}, P_{B,j}$ and $P_{C,j}$.

A simple but key observation is the following when applied to the Guruswami--Sudan interpolation polynomial $Q$. 
\begin{restatable}{observation}{keyobservation}\label{obs:intro-list-membership}
Let $Q$ be the Guruswami--Sudan interpolation polynomial and $f\in L$. Then:
\begin{enumerate}
  \item If $f(\alpha_j)=\beta_j$ for some $j$, then $(Y-f(X))$ divides $Q_{A,j}$.
  \item If $Q=Q_{A,j}$ for some $j$, then $f(\alpha_j)=\beta_j$.
  \item If $f(\alpha_j)\neq\beta_j$ for some $j$, then $(Y-f(X))$ divides $Q_{B,j}$
\end{enumerate}
\end{restatable}
The first item above follows from the definition of $Q_{A,j}$ and similarly the third item from the definition of $Q_{B,j}$. For the second item, we recall that $(Y-f(X))$ divides $Q(X,Y)$ for every $f$ in the list $L$. As a consequence, $(Y- f(\alpha_j))$ must divide $Q(\alpha_j, Y)$, which by the definition of $Q_{A_j}$ and the fact that $Q = Q_{A,j}$, only has irreducible factors of the form $(Y - \beta_j)$. So, $f(\alpha_j)$ must equal $\beta_j$.

Thus, if we can produce a bivariate $P$ divisible by $(Y-f(X))$ and satisfying $P=P_{A,j}$ for at least $k$ indices $j$, then we can recover $f$ by interpolating on those $(\alpha_j,\beta_j)$.

Our algorithm (\cref{thm:intro-main}) is an efficient deterministic realization of this plan, built from three components.

\paragraph*{(1) Deterministic local splitting $P\mapsto (P_{A,j},P_{B,j},P_{C,j})$.}
\begin{theorem}[monic version of \cref{thm:split}]\label{thm:intro-split}
Let $\F$ be any finite field. There is a deterministic algorithm \textsc{Split} which, given a polynomial $P\in\F[X,Y]$ that is monic in the variable $Y$ and a point $(\alpha,\beta)\in\F^2$, outputs polynomials $P_{A(\alpha,\beta)}$, $P_{B(\alpha,\beta)}$ in time $\poly(\deg P,\log|\F|)$.
\end{theorem}
The method follows the Hensel-lifting paradigm but avoids randomized univariate factorization by using $(\alpha_j,\beta_j)$ directly. First write
\[
P(\alpha_j,Y)=(Y-\beta_j)^{m_j}\cdot \hat P_j(Y),\qquad \text{where }\hat P_j(\beta_j)\neq 0,
\]
computed deterministically by repeated division by $(Y-\beta_j)$. If $m_j=0$ then $P=P_{B,j}$; if $m_j=\deg_Y P(\alpha_j,Y)$ then $P=P_{A,j}$; otherwise the coprime factorization $\big((Y-\beta_j)^{m_j},\widehat P_j(Y)\big)$ serves as the base case for Hensel lifting to recover $P_{A,j},P_{B,j},P_{C,j}$ deterministically. While the standard application of Hensel lifting (c.f., \cite{KoppartySS2015,SinhababuT2021}) uses randomized univariate factorization to seed it, the above deterministic factoring of $P(\alpha_j,Y)$ into coprime factors $(Y-\beta_j)^{m_j}$ and $\hat{P_j}(Y)$ circumvents the use of randomness in our application of Hensel lifting. See \autoref{sec:split} for details.

\paragraph*{(2) Refinement to a stable set.}
Call $P$ \emph{stable} (with respect to the received word $w$) if for every $j\in[n]$, either $P=P_{A,j}$ or $P=P_{B,j}$. Initialize $S\gets\{Q\}$. While some $P\in S$ is not stable (i.e., there is a $j \in [n]$ such that $P \notin \{P_{A_j},P_{B_j}\}$), update
\[
S \;\leftarrow\; \Big( S \cup \{P_{A,j},P_{B,j}\} \Big)\setminus\{P\}.
\]
This refinement uses \cref{thm:intro-split} at each step.
 
\paragraph*{(3) Decoding from stability and efficiency.}
If $S$ is stable and $(Y-f)\mid P\in S$, then by \cref{obs:intro-list-membership}, every $j$ with $P=P_{A,j}$ corresponds to a true agreement $f(\alpha_j)=\beta_j$. We prove that after stabilization, for every $f\in L$ there exists $P\in S$ such that $(Y-f)\mid P$ and $P=P_{A,j}$ holds on at least $k$ indices $j$. Interpolating on these points recovers $f$, and iterating over $P\in S$ outputs all $f\in L$.

To conclude the argument, we need to show that the refinement to a stable set can be done in polynomial time. We do this via a potential function argument. See \autoref{sec:GS} for the full details and analysis.
\section{Preliminaries}

For a natural number $n$, we use $[n]$ to denote the set $\{1, 2, \ldots, n\}$. We say that a function $T(n)$ is $\poly(n)$ if there is a constant $c$ (independent of $n$) such that for all sufficiently large $n$, $T(n) \leq n^c$. 

Throughout the paper, $\F$ denotes a finite field of size $q=p^r$, for some prime $p$ and positive integer $r$. The prime $p$ is referred to as \emph{characteristic} of the field $\F$. We will be working with the ring $\F[X,Y]$ of bivariate polynomials over the field $\F$. 

\subsection{Resultants and GCD}

\paragraph*{Greatest common divisor:} The \emph{greatest common divisor (GCD)} of two polynomials $P, Q \in \F[X,Y]$, denoted by $\gcd(P,Q)$ is a polynomial $D\in \F[X,Y]$  that divides $P$ and $Q$, and such that every common divisor of $P$ and $Q$ also divides $D$. If $D = \gcd(P,Q)$, there exist two polynomials $A, B \in \F[X,Y]$, referred to as the \emph{Bezout coefficients} such that 
\[ A\cdot P + B\cdot Q = D .\]

\paragraph*{Resultant:} Let $g=\sum_{i=0}^{d_1}g_i(X)Y^i$ and $h=\sum_{i=0}^{d_2}h_i(X)Y^i$ be polynomials in $\F[X,Y]$ of $Y$-degree exactly $d_1$ and $d_2$. Denote the vector space $W_{d}=\{f\in \F[X,Y]\mid \deg_Y(f)<d\}$ over $\F(X)$. Then consider linear map $\psi:W_{d_2}\times W_{d_1}\to W_{d_1+d_2}$ where $\psi(a(X,Y),b(X,Y))=a(X,Y)\cdot g(X,Y)+b(X,Y)\cdot h(X,Y)$ where $a,b\in\F[X,Y]$ and $\deg_Y(a)<d_2$, $\deg_Y(b)<d_1$. Then the matrix of the $\psi$, over $\F[X]$ is called the \emph{Sylvester Matrix} with respect to $Y$. And the determinant the Sylvester matrix is called the \emph{Resultant} of $g,h$ with respect to $Y$, denoted by $\Res_Y(g,h)$

\begin{lemma}[Resultant Properties]\label{lem:res}
Let $d_1,d_2\geq 1$. Let $g(X,Y),h(X,Y)\in\F[X,Y]$ have $Y$-degree exactly $d_1,d_2$ respectively.  Then the following are true:
\begin{enumerate}
    \item If each $g_i$ is a polynomial of degree at most $e_1$ and each $h_j$ is a polynomial of degree at most $e_2$ then $\deg(\Res_Y(g,h))\leq 2e_1e_2$ 
    \item $\Res_Y(g,h)=0$ if and only if $g,h$ have a nontrivial GCD in the ring $\F(X)[Y]$.
    \item There exist polynomials $a(X,Y)$ and $b(X,Y)$ in $\F[X,Y]$ such that: \[ \Res_Y(g,h)= a\cdot g+b\cdot h\]
\end{enumerate}
\end{lemma}
The following lemma is an easy consequence of Gauss' lemma and will be useful for us. 
\begin{lemma}[Gauss' lemma]\label{lem:gauss}
Let $A(X,Y) = \sum_{i = 0}^d A_i(X)Y^i$ and $B(X,Y) = \sum_{j = 0}^e B_j(X)Y^j$ be polynomials in $\F[X,Y]$ of $Y$-degrees at least one. 

Then, $A$ and $B$ have a GCD with $Y$-degree at least one in the ring $\F(X)[Y]$ if and only if they have GCD with $Y$-degree at least one in $\F[X, Y]$. 
\end{lemma}

\subsection{Hasse derivatives and multiplicties}
In many parts of the proof, especially in the context of the list decoding algorithm of Guruswami--Sudan, we have to work with the notion of multiplicities. We start by recalling the notion of Hasse derivatives and multiplicity. 

\begin{definition}[Hasse derivatives]
Let $Q(X,Y) \in \F[X,Y]$ be a bivariate polynomial. Then for any $(e_1, e_2) \in \Z_{\geq 0}^2$, the $(e_1, e_2)$-th Hasse derivative of $Q$, denoted by $Q^{(e_1, e_2)}(X,Y)$ is defined to be the coefficient of the monomial $Z_1^{e_1}Z_2^{e_2}$ in the polynomial $Q(X+Z_1, Y+Z_2)$, when it is viewed as a bivariate in $Z_1, Z_2$. 

$(e_1+e_2)$ is said to be the \emph{order} of the Hasse derivative $Q^{(e_1, e_2)}$. 

\end{definition}

\begin{definition}[multiplicity]
Let $Q(X,Y) \in \F[X,Y]$ be a bivariate polynomial, $m$ be a non-negative integer and $(\alpha, \beta) \in \F^2$ be a point. We say that $Q$ vanishes with multiplicity at least $m$ on on  $(\alpha, \beta)$ if every Hasse derivative of $Q$ of order at most $(m-1)$ evaluates to zero at $(\alpha, \beta)$.
\end{definition}

The definitions of both Hasse derivatives and multiplicity extend naturally to polynomials with an arbitrary number of variables. Throughout the paper, these are invoked only for bivariate or univariate polynomials.

\subsection{Computing with polynomials:} In this section, we mention several standard \emph{deterministic} building blocks for working with polynomials.

\paragraph*{Exact division:} There exists a deterministic algorithm that on input two polynomials $P,Q\in \F[X,Y]$ with the guarantee that $Q$ divides $P$, finds the quotient $P/Q$ in time $\poly(\deg P, \deg Q, \log |\F|)$.

\paragraph{GCD Computation:} Euclid's algorithm computes the GCD of $n$ univariate polynomials $P_1,\dots,P_n \in \F[X]$ in time $\poly(\max \deg P_i, \log |\F|)$. The extended-Euclid algorithm, on input two polynomials $P,Q \in \F[X]$ computes the $\gcd(P,Q)$ as well as the Bezout coeffients in $\poly(\deg P, \deg Q, \log |\F|)$ time. 

\paragraph{Content of polynomials.} Given a polynomial $P \in \F[X,Y]$, we can view it as a polynomial in $Y$ with coefficients from the ring $\F[X]$, more precisely write $P(X,Y) = \sum_\ell P^{(\ell)}(X)\cdot Y^\ell$. The $X$-content of the polynomial $P$ is $\gcd(P^{(1)},P^{(2)},\dots)$. 

\paragraph{Polynomial factorization for small characteristic fields:} As mentioned in the introduction, Shoup's algorithm \cite{Shoup1990} computes the factorization of degree $d$ bivariate polynomials in deterministic $\sqrt{p}\cdot \poly(d, \log |\F|)$ time where $p$ is the characteristic of the field.

For fields of larger characteristic, we will use the following standard tools from computational algebra: Newton's iteration and Hensel lifting. 

\paragraph*{Newton's iteration.} As mentioned in the proof overview (\cref{sec:overview}), we will be using Newton's iteration repeatedly to find the $Y$-roots of bivariate polynomials. 

\NI*

\paragraph*{Hensel lifting.} Hensel lifting allows us to lift a local factorization of a polynomial modulo a power of $(X-\alpha)$ to a higher power.  
This may be viewed as the analog of Newton iteration for factoring, as opposed to root lifting.  We give more details of the Hensel lifting in \cref{sec:hensel}.

\paragraph*{A note on the characteristic of the underlying field.}
Typical applications of Newton's iteration and Hensel's lifting assume that the characteristic of the underlying field is at least the degree of the polynomials involved. While it is possible to use these ideas for polynomial factorization over fields of small characteristic with some care, it will be convenient for us to assume that the characteristic is at least $\poly(n)$ for our applications, since anyway when the characteristic is smaller one can use Shoup's deterministic factorization to obtain deterministic list-decoding algorithms.
    
\section{Derandomizing Sudan's decoder}\label{sec:sudan}
In this section, we prove \cref{thm:intro-sudan}. Recall that Sudan's list-decoding algorithm consists of two main steps:

\begin{enumerate}
  \item \textbf{Interpolation Step:}  
  Given the received word $w=\{(\alpha_j,\beta_j)\}_{j\in[n]}$, find a nonzero bivariate polynomial $Q(X,Y)\in\F[X,Y]$ of $(1,k-1)$-degree at most $D$ such that $Q(\alpha_j,\beta_j)=0$ for all $j\in[n]$.  
  Here $D$ is a parameter depending on $n$ and $k$, set to $D=\sqrt{2(k-1)n}$.

  \item \textbf{Factorization Step:}  
  Find all factors of $Q(X,Y)$ of the form $Y-f(X)$ where $\deg(f)\le k-1$, such that $f$ agrees with $w$ on more than $\sqrt{2(k-1)n}$ coordinates.
\end{enumerate}

The first step can be carried out deterministically in polynomial time by solving a homogeneous linear system: there are more variables than equations, guaranteeing a nontrivial solution.  
Thus, the only source of randomness lies in the \emph{factorization step}.  
The main technical result of this section is the following theorem, which, combined with the above observation, immediately implies \cref{thm:intro-sudan}.

\begin{theorem}\label{thm:sudan-technical}
There is a deterministic algorithm that, for every finite field $\mathbb{F}$ and parameters $n,k\in\mathbb{N}$, runs in time $\poly(n,\log|\mathbb{F}|)$ and, given as input a received word $w=\{(\alpha_j,\beta_j)\}_{j\in[n]}$ and a polynomial $Q(X,Y)\in\mathbb{F}[X,Y]$ of $(1,k-1)$-degree at most $\sqrt{2(k-1)n}$ satisfying $Q(\alpha_j,\beta_j)=0$ for all $j\in[n]$, returns all polynomials $f\in\mathbb{F}[X]$ of degree at most $(k-1)$ that agree with $w$ on more than $\sqrt{2(k-1)n}$ coordinates.
\end{theorem}

We now describe and analyze the deterministic algorithm.

\subsection{The algorithm}

\begin{algorithm}[H]
\caption{\textsc{Deterministic-Sudan}}
\DontPrintSemicolon
\SetAlgoNoEnd
\KwIn{$w=\{(\alpha_j,\beta_j)\}_{j\in[n]}$ and $Q(X,Y)\in \F[X,Y]$ of $(1,k-1)$-degree at most $\sqrt{2(k-1)n}$ such that $Q(\alpha_j,\beta_j)=0$ for all $j\in[n]$}
\KwOut{All $f\in\F[X]$ of degree at most $(k-1)$ that agree with $w$ on more than $\sqrt{2(k-1)n}$ coordinates}
\nonl\hrulefill\\

\lIf{there exists $f(X)$ of degree $\le k-1$ such that $f(\alpha_j)=\beta_j$ for all $j\in[n]$}{
  \Return{$\{f\}$}
}

\If{$\deg_Y(Q)=1$}{
  Write $Q(X,Y)=Q_0(X)+Y\cdot Q_1(X)$\;
  Set $f(X)\gets -\frac{Q_0(X)}{Q_1(X)}$\;
  \lIf{$\deg(f)\le k-1$ and $f$ has agreement with $w$ exceeding $\sqrt{2(k-1)n}$}{
    \Return{$\{f\}$}
  }
}

Set $L\gets\emptyset$\;
Set $P\gets Q^{(0,1)}(X,Y)$\;
\While{$\deg_Y(P)>0$ and there exists $j\in[n]$ such that $P(\alpha_j,\beta_j)\neq0$}{
  
  $\varphi_0\gets \beta_j$\;
  \For{$i=1$ to $(k-1)$}{
    Use Newton Iteration (\cref{lem:NI}) to compute the lift $\varphi_i$ from $\varphi_{i-1}$ satisfying  
    $Q(X,\varphi_i(X))\equiv 0\pmod{\langle X-\alpha_j\rangle^{i+1}}$ and  
    $\varphi_i\equiv \varphi_{i-1}\pmod{\langle X-\alpha_j\rangle^{i}}$
  }
  Set $\varphi\gets \varphi_{k-1} \bmod \langle X-\alpha_j\rangle^k$\;
  \lIf{$\varphi$ agrees with $w$ on more than $\sqrt{2(k-1)n}$ coordinates}{
    $L\gets L\cup\{\varphi\}$
  }
  Set $w \gets w \setminus \{(\alpha_j,\beta_j)\}$\;
}
$L'\gets \textsc{Deterministic-Sudan}(P,w)$\;
\Return{$L\cup L'$}
\end{algorithm}

\subsection{Correctness of the algorithm}
The theorem is interesting only when the desired agreement threshold, $\sqrt{2(k-1)n}+1$, is smaller than $n$.  
If it equals $n$, the problem reduces to simple polynomial interpolation: we can directly check whether there exists $f$ of degree at most $k-1$ such that $f(\alpha_j)=\beta_j$ for all $j\in[n]$.

Henceforth, assume $\deg_Y(Q)\ge1$.  
If $\deg_Y(Q)=0$, then $Q$ is a univariate polynomial in $X$ of degree at most $\sqrt{2(k-1)n}$ vanishing at $n$ distinct points, which would imply $Q\equiv0$, a trivial case.

We begin with a simple observation.

\begin{claim}\label{clm:sudan-correctness-1}
Let $f(X)$ be a polynomial of degree $<k$ that agrees with $w$ on more than $\sqrt{2(k-1)n}$ coordinates.  
Let $R(X,Y)$ be a nonzero polynomial of $(1,k-1)$-degree at most $\sqrt{2(k-1)n}$ such that $R(\alpha_j,\beta_j)=0$ for every $j$ where $f(\alpha_j)=\beta_j$.  
Then $R(X,f(X))$ is identically zero.
\end{claim}

\begin{proof}
The univariate $R(X,f(X))$ has degree at most the $(1,k-1)$-degree of $R$, namely $\sqrt{2(k-1)n}$.  
It vanishes at each $\alpha_j$ where $f(\alpha_j)=\beta_j$, i.e., on more than $\sqrt{2(k-1)n}$ distinct points.  
Hence $R(X,f(X))$ has more zeros than its degree and must be identically zero.
\end{proof}

By \cref{clm:sudan-correctness-1}, for any such $f$, we have $Q(X,f(X))\equiv0$.  
We now prove by induction on $\deg_Y(Q)$ that every such $f$ is included in the list $L$.

\paragraph{Base case:}
If $\deg_Y(Q)=1$, there is at most one polynomial $f$ such that $Q(X,f(X))\equiv0$, and the algorithm explicitly checks for this case.

\paragraph{Induction step:}
Assume correctness for all polynomials of $Y$-degree at most $t$.  
Let $Q$ have $Y$-degree $t+1>1$.  
We consider two cases:

\textbf{Case I:} There exists $j$ such that $f(\alpha_j)=\beta_j$ and $Q^{(0,1)}(\alpha_j,\beta_j)\ne0$.  
Then the conditions of the Newton Iteration lemma (\cref{lem:NI}) are met, and the algorithm successfully lifts $\beta_j$ modulo successive powers of $\langle X-\alpha_j\rangle$ to recover $f$.  
Since $f$ has sufficiently large agreement, it is added to $L$.

\textbf{Case II:} For every $j$ with $f(\alpha_j)=\beta_j$, we also have $Q^{(0,1)}(\alpha_j,\beta_j)=0$.  
In this case, $Q^{(0,1)}(X,Y)$ is a nonzero polynomial (since $\deg_Y(Q)>1$) of $(1,k-1)$-degree at most $\sqrt{2(k-1)n}$ that vanishes on all agreement points of $f$.  
By \cref{clm:sudan-correctness-1}, $Q^{(0,1)}(X,f(X))\equiv0$, so $f$ is a $Y$-root of $Q^{(0,1)}(X,Y)$.  
By the induction hypothesis, the recursive call with $Q^{(0,1)}$ includes $f$ in $L$.

This completes the proof of correctness.
\qedhere

\subsection{Running time analysis}
In each recursive call, the $Y$-degree of $Q$ decreases by at least one.  
Hence, the recursion depth is at most $\deg_Y(Q)=O(\sqrt{2(k-1)n})$, and only one recursive call is made per level.  
Each call involves polynomial-time computations; evaluating Hasse derivatives, performing polynomial divisions, verifying agreement, and executing $\poly(n)$ many calls to the Newton Iteration procedure (\cref{lem:NI}), each of which requires $\poly(n)$ field operations.

Therefore, the total running time of the algorithm is $\poly(n,\log|\F|)$, and the algorithm is deterministic, completing the proof of \cref{thm:sudan-technical}. 

\section{Derandomizing the Guruswami-Sudan decoder}\label{sec:GS}
In this section, we prove \cref{thm:intro-main}.  
We begin by recalling the high-level structure of the Guruswami--Sudan algorithm~\cite{GuruswamiS1999}, which proceeds in two stages.
\begin{enumerate}
  \item \textbf{Interpolation.}  
  Given the received word $w=\{(\alpha_j,\beta_j)\}_{j\in[n]}$, find a nonzero bivariate polynomial $Q(X,Y)$ of $(1,k-1)$-degree at most $D$ such that each point $(\alpha_j,\beta_j)$ is a root of $Q$ with multiplicity at least $m$.  
  To decode from agreement $\sqrt{n(k-1)}$, set $m=2\sqrt{n(k-1)}$ and $D=\Theta(m\sqrt{(k-1)n})$.
  
  \item \textbf{Factorization.}  
  Find all factors of $Q(X,Y)$ of the form $(Y-f(X))$, where $\deg f \le k-1$ and $f$ agrees with $w$ on at least $\sqrt{(k-1)n}$ coordinates.
\end{enumerate}
The following two lemmas from~\cite{GuruswamiS1999} ensure that the interpolation step succeeds and that all sufficiently close codewords appear as $Y$-roots of $Q$.

\begin{lemma}[interpolation guarantees \cite{GuruswamiS1999}]\label{constant-y-degree}
For every $n > k \in \mathbb{N}$, there exist parameters $m = 2 \sqrt{n(k-1)}$ and $D = \Theta(m\sqrt{(k-1)n})$ such that the following hold.

For every received word $w = \{(\alpha_j, \beta_j)\}_{j \in [n]}$, there exists a nonzero polynomial $Q(X,Y)$ satisfying:
\begin{enumerate}
  \item $(1,k-1)$-degree$(Q) \le D$,
  \item $\deg_Y(Q) = O(m\sqrt{n/(k-1)})$,
  \item For every $j \in [n]$, $Q$ vanishes at each $(\alpha_j, \beta_j)$ with multiplicity at least $m$.
\end{enumerate}
Moreover, such a $Q$ can be found deterministically in $\mathrm{poly}(n,\log |\F|)$ time.
\end{lemma}

\begin{lemma}[root containment \cite{GuruswamiS1999}]\label{lem:close-enough-polynomials-are-roots-GS}
Let $Q$ satisfy the properties in \cref{constant-y-degree}.  
Then for every polynomial $f$ of degree at most $k-1$ that agrees with $w$ on at least $\sqrt{(k-1)n}$ coordinates, we have $Q(X,f(X)) \equiv 0$, i.e., $(Y-f(X))$ divides $Q(X,Y)$.
\end{lemma}
\begin{proof}
The proof follows immediately from the bounds on $D$ and $m$ in \cref{constant-y-degree} and the observation that at each agreement point of $f$ and $w$, the univariate $Q(X,f(X))$ vanishes with multiplicity at least $m$.
\end{proof}

Thus, the list-decoding task reduces to finding all $Y$-roots of $Q(X,Y)$ of degree at most $(k-1)$.  
The main challenge is to do this deterministically in time $\mathrm{poly}(n,\log|\mathbb{F}|)$.  
The next theorem, which is the main contribution of this paper, provides such an algorithm.
\begin{theorem}\label{thm:det-gs}
There is a deterministic algorithm such that for every finite field $\mathbb{F}$ and parameters $n>k\in\mathbb{N}$:
\begin{itemize}
  \item It takes as input a received word $w=\{(\alpha_j,\beta_j)\}_{j\in[n]}$ and a polynomial $Q(X,Y)\in\mathbb{F}[X,Y]$ with $\deg_Y(Q)=\poly(n)$, under the promise that every degree-$(k-1)$ polynomial with agreement at least $\sqrt{(k-1)n}$ with $w$ is a $Y$-root of $Q$.
  \item It outputs all such polynomials $f$ in time $\mathrm{poly}(n,\log|\mathbb{F}|)$.
\end{itemize}
\end{theorem}

The rest of this section proves this theorem, completing the proof of \cref{thm:intro-main}. 

Let $L$ denote the set of polynomials $f(X)$ of degree less than $k$ and agreement more than $\sqrt{(k-1)n}$ with $w$.  
Factor $Q(X,Y)=\prod_{i=1}^s Q_i(X,Y)$, where each $Q_i$ is  (with multiplicity) and has $\deg_Y(Q_i)\ge1$.  
We may assume without loss of generality that $Q$ has no factors depending only on $X$, since dividing out the GCD of its pure-$X$ components does not affect the property that each $f\in L$ is a $Y$-root of $Q$.

In the rest of this section we prove this theorem, and this would complete the proof of \cref{thm:intro-main}. We start by recalling some notation already discussed in the proof overview. Let $L$ be the set of polynomials $f(X)$ of degree less than $k$ and agreement more than $\sqrt{(k-1)n}$ with the received word $w$. 

\subsection{Local Splitting}

Let $P(X,Y)\in\mathbb{F}[X,Y]$ be a bivariate polynomial with no pure-$X$ factors and a point $(\alpha,\beta)\in\mathbb{F}^2$. Given the factorization $P=\prod_{i=1}^s P_i$ into irreducibles (with multiplicity), define:
\begin{align*}
  A(\alpha,\beta) &:= \{i \in [s] \mid P_i(\alpha,Y)=\gamma\,(Y-\beta)^{m_i} \text{ for some } m_i\ge 1,\ \gamma\in\F \setminus \{0\}\},\\
  B(\alpha,\beta) &:= \{i \in [s]\mid  P_i(\alpha,\beta)\neq 0, \deg P_i(\alpha,Y)\ge 1\},\\
  C(\alpha,\beta) &:= \{i \in [s] \mid P_i(\alpha,Y)=(Y-\beta)^{m_i}\cdot \hat P_i(Y) \text{ for some } m_i\ge 1, \deg_Y \widehat P_i\ge 1, \widehat{P}_i(\beta)\neq 0\},\\
D(\alpha,\beta) &:= \{i \in [s]\mid  P_i(\alpha,\beta)\neq 0, \deg_Y P_i(\alpha, Y) = 0\}.
\end{align*}
Let $P_{A(\alpha,\beta)}=\prod_{i\in A(\alpha,\beta)}P_i$, and define $P_{B(\alpha,\beta)}$, $P_{C(\alpha,\beta)}$, $P_{D(\alpha,\beta)}$ analogously.  
For each $(\alpha_j,\beta_j)$ in $w$, we write $A_j,B_j,C_j, D_j$ and $P_{A,j},P_{B,j},P_{C,j}, P_{D,j}$ for brevity instead of $A(\alpha_j,\beta_j), B(\alpha_j,\beta_j), C(\alpha_j,\beta_j)$, $ D(\alpha_j,\beta_j)$ and $P_{A(\alpha_j,\beta_j)}, P_{B(\alpha_j,\beta_j)}, P_{C(\alpha_j,\beta_j)}, P_{D(\alpha_j,\beta_j)}$ for $(\alpha_j,\beta_j)$.

We recall that in \autoref{sec:overview}, we focused on bivariates that are monic in $Y$. For such polynomials, all the irreducible factors can be assumed to be monic in $Y$ without loss of generality. As a consequence, the $Y$-degree of any of its irreducible factors cannot drop when $X$ is set to $\alpha \in \F$ and hence the set $D$ defined above is empty. However, the Guruswami-Sudan interpolation polynomial $Q$ need not be monic in $Y$ and we need to exercise some additional care to handle such polynomials. 

We now apply this definition to Guruswami--Sudan interpolation polynomial $Q$ (or more precisely $Q$ with its pure $X$ factors shaved off). 
If $f\in L$, then $Q(X,f(X))\equiv 0$ or equivalently $Y-f(X)\mid Q$. Now for any $j\in[n]$ if $f(\alpha_j)=\beta_j$ then $Y-f(X)\equiv Y-\beta_j\bmod{\langle X-\alpha_j\rangle}$. Hence, $Y-f(X)\mid Q_{A,j}(X,Y)$. And if $f(\alpha_j)\neq \beta_j$ then $(Y-f(X))(\alpha_j,\beta_j)=\beta_j-f(\alpha_j)\neq 0$. Therefore, $Y-f(X)\mid Q_{B,j}(X,Y)$. Therefore, we have the following observation:
\begin{observation}\label{obs:list-membership-category}
For every $f\in L$:
\begin{enumerate}
  \item If $f(\alpha_j)=\beta_j$ for some $j$, then $(Y-f(X))\mid Q_{A,j}(X,Y)$.
  \item If $f(\alpha_j)\ne\beta_j$ for some $j$, then $(Y-f(X))\mid Q_{B,j}(X,Y)$.
\end{enumerate}
\end{observation}
We also rely on the following simple observation in the description of our algorithm. 

\begin{observation}\label{obs:uselss}
Let $g(X,Y) \in \F[X,Y]$ be a bivariate polynomial such that there is an $\alpha \in \F$ for which $g(\alpha, Y)$ is a non-zero field constant. Then, $g(X,Y)$ is not divisible any bivariate polynomial that is monic in $Y$. 

In particular, $g(X,Y)$ does not have a factor of the form $(Y-f(X))$. 
\end{observation}

\begin{proof}
    We will prove the contrapositive -  we assume that $g(X,Y)$ has a factor $P(X,Y)$ that is monic in $Y$ and show that $g(\alpha, Y)$ cannot be a non-zero field constant.  

    Let $R(X,Y)$ be a polynomial that is monic in $Y$ (and hence has $Y$-degree at least one) and divides $g(X,Y)$. Thus, there is a bivariate $R'$ such that $g(X,Y) = R(X,Y)R'(X,Y)$. Now, by setting $X$ to $\alpha$ in the above identity, we get $g(\alpha,Y) = R(\alpha,Y)R'(\alpha,Y)$. Since $g(\alpha, Y)$ is non-zero, both $R'(\alpha, Y)$ and $R(\alpha, Y)$ are non-zero polynomials. Moreover,  $R$ was monic in $Y$, so its degree in $Y$ does not reduce when $X$ is set to $\alpha$. But this implies that the $g(\alpha, Y)$ has degree at least as $R(\alpha, Y)$, which is at least one. 
\end{proof}
This observation motivates the following definition. 

\begin{definition}[Useless bivariates]\label{def:useless}
A bivariate $g(X,Y) \in \F[X,Y]$ is said to be useless with respect to a set $E \subseteq \F$ if there exists an $\alpha \in E$ such that $g(\alpha,Y)$ is a non-zero field constant. 
\end{definition}
We now define the notion of stability, that is once again crucially used in our algorithm. 
\begin{definition}[stability of a polynomial]\label{def:stability}
    We call a polynomial $P(X,Y)$ \emph{stable} with respect to the point $(\alpha,\beta)$ if $\deg P(\alpha,Y) \geq 1$ and one of the following holds
    \begin{itemize}
    \item $P(\alpha,Y) = \gamma (Y-\beta)^r$ for some $\gamma \in \F\setminus\{0\}$ and positive integer $r$
    \item $P(\alpha,\beta) \neq 0$
    \end{itemize}
    
    We call a polynomial $P(X,Y)$ \emph{stable} with respect to the received word $w$ if for all $j\in[n]$, $P$ is stable with respect to $(\alpha_j,\beta_j)$. We will usually drop the phrase "with respect to the received word $w$"

    A finite set $S\subseteq\mathbb{F}[X,Y]$ is stable if every $P\in S$ is stable.
\end{definition}

Our algorithm relies on a deterministic subroutine, \textsc{Split}, that given $P(X,Y)$ and $(\alpha,\beta)$ computes $P_{A(\alpha,\beta)}R_1$ and $P_{B(\alpha,\beta)}R_2$ such that $R_1\cdot R_2$ divides $P_{D(\alpha,\beta)}$ in time $\poly(\deg(P),\log|\mathbb{F}|)$.  
We defer its description to \cref{sec:split} and restate its guarantee below.

\begin{restatable}[\Split\ procedure]{theorem}{splitthem}\label{thm:split}
    Let $\F$ be any finite field. There is a deterministic algorithm \textsc{Split} which, given $P\in\F[X,Y]$ and $(\alpha,\beta)\in\F^2$, outputs polynomials $(P_1,P_2)$ such that $P_1= P_{A(\alpha,\beta)}\cdot R_1$ and $P_2= P_{B(\alpha,\beta)}\cdot R_2$ where $R_1 \cdot R_2$ divides $P_{D(\alpha,\beta)}$.

    Moreover, if $P$ is stable with respect to $(\alpha,\beta)$, \Split\ outputs either $(P,1)$ or $(1,P)$ depending on which of the two conditions in \cref{def:stability} is met.
\end{restatable}

The \emph{moreover} part of the above theorem explains the behavior of the algorithm \Split\ on inputs that are stable. The corollary below summarises its properties on inputs that are not stable. This will be useful for us in the analysis of the algorithm later in this section. 

\begin{corollary}\label{cor:split-on-unstable-inputs}
Let $P(X,Y)$ be a polynomial such that the degree of $P(\alpha, Y)$ is at least one and $P$ is not stable at the point $(\alpha, \beta)$. Let $(P_1, P_2)$ be the output of the algorithm \Split\ on inputs $P, (\alpha, \beta)$. 

Then, $\deg_Y(P_1), \deg_Y(P_2)$ are both strictly smaller than $\deg_Y(P)$. Moreover, 
$\deg_Y(P_1) + \deg_Y(P_2) \leq \deg_Y(P)$. 
\end{corollary}
\begin{proof}
    The \emph{moreover} part of the corrollary follows immediately from the structure and guarantees on  $P_1, P_2$ in \autoref{thm:split}. So, we focus on proving that the $Y$-degrees of both $P_1$ and $P_2$ is strictly smaller than the $Y$-degree of $P$. 

    Since $\deg_Y(P(\alpha, Y))$ is at least one, and $P$ is not stable at $(\alpha, \beta)$, we get from \autoref{def:stability} that $P(\alpha, \beta)$ is zero and $P(\alpha, Y)$ is not a non-zero scalar multiple of a power of $(Y-\beta)$. In other words, there is a univariate $\hat{P} \in \F[Y]$ of degree at least one, a positive integer $r$ such that 
    \[
    P(\alpha, Y) = (Y - \beta)^r \cdot \hat{P}(Y) ,
    \]
    and $\hat{P}(\beta)$ is non-zero.  

 From the above and the definitions of $P_{A(\alpha, \beta)}, P_{B(\alpha, \beta)}, P_{C(\alpha, \beta)}, P_{D(\alpha, \beta)}$, we get that the product $P_{B(\alpha, \beta)}(\alpha, Y)\cdot P_{C(\alpha, \beta)}(\alpha, Y)$ must be divisible by $\hat{P}$, and hence has degree at least one. Furthermore, from \autoref{thm:split}, we have that $P_1$ divides $P_{A(\alpha, \beta)}P_{D(\alpha, \beta)}$. Hence, we have the following sequence of inequalities.  
 \begin{align*}
 \deg_Y(P_1) & \leq \deg_Y(P_{A(\alpha, \beta)}(X,Y)) + \deg_Y(P_{D(\alpha, \beta)}(X,Y)) \\
            & = \deg_Y(P) - (\deg_Y(P_{B(\alpha, \beta)}(X,Y)) + \deg_Y(P_{C(\alpha, \beta)}(X,Y))) \\
            & \leq \deg_Y(P) - (\deg(P_{B(\alpha, \beta)}(\alpha,Y)) + \deg(P_{C(\alpha, \beta)}(\alpha,Y))) \\
            & \leq \deg_Y(P) -1 .
 \end{align*}
 Here, we used the fact that the $Y$ degree cannot increase when $X$ is set to $\alpha$ and the fact that $(\deg(P_{B(\alpha, \beta)}(\alpha,Y)) + \deg(P_{C(\alpha, \beta)}(\alpha,Y)))$ is equal to the degree of the product $P_{B(\alpha, \beta)}(\alpha, Y)\cdot P_{C(\alpha, \beta)}(\alpha, Y)$, which as discussed earlier is at least one. 
 
 Similarly, since $r$ is positive, we get that the product $P_{A(\alpha, \beta)}(\alpha, Y)\cdot P_{C(\alpha, \beta)}(\alpha, Y)$ must be divisible by $(Y-\beta)^r$, and hence has degree at least one.  We again have a sequence of inequalities. 
 \begin{align*}
 \deg_Y(P_2) & \leq \deg_Y(P_{B(\alpha, \beta)}(X,Y)) + \deg_Y(P_{D(\alpha, \beta)}(X,Y)) \\
            & = \deg_Y(P) - (\deg_Y(P_{A(\alpha, \beta)}(X,Y)) + \deg_Y(P_{C(\alpha, \beta)}(X,Y))) \\
            & \leq \deg_Y(P) - (\deg(P_{A(\alpha, \beta)}(\alpha,Y)) + \deg(P_{C(\alpha, \beta)}(\alpha,Y))) \\
            & \leq \deg_Y(P) -1 .\qedhere
 \end{align*}
\end{proof}

\subsection{Details of the algorithm}
We are now ready to describe the deterministic variant of the Guruswami--Sudan algorithm.  
Given $Q(X,Y)$, we first remove any pure-$X$ factors to obtain $Q'$.  
The algorithm then proceeds in two stages:

\begin{enumerate}
  \item \textbf{Refinement.}  
  Initialize $S\leftarrow\{Q'\}$.  
  Repeatedly apply \textsc{Split} to every $P\in S$ and each $(\alpha_j,\beta_j)$, replacing $P$ by its components $P_{A,j},P_{B,j}$, until $S$ becomes stable.
  
  \item \textbf{Extraction.}  
  For each $g\in S$, let $U=\{j : g(\alpha_j,\beta_j)=0\}$.  
  If there exists a degree-$(k-1)$ polynomial $f$ that agrees with $w$ on all indices in $U$, output $f$.
\end{enumerate}
The formal description of the algorithm follows.

\begin{algorithm}[H]
\caption{\textsc{Deterministic-Factorization-for-List-Decoding}}\label{algo:det-gs-algo}
\DontPrintSemicolon
\SetAlgoNoEnd

\KwIn{Received word $w=\{(\alpha_j,\beta_j)\}_{j\in[n]}$ and interpolation polynomial $Q(X,Y)\in\F[X,Y]$ such that $Q(\alpha_j,\beta_j)=0$ for all $j\in[n]$.}
\KwOut{All polynomials $f\in\F[X]$ of degree at most $(k-1)$ that agree with $w$ on at least $\sqrt{(k-1)n}$ coordinates.}
\nonl\hrulefill \\

\BlankLine
\nonl \textbf{Step 1: Remove pure-$X$ factors.} \;
Write $Q(X,Y) = \sum_{\ell} Q^{(\ell)}(X) Y^\ell$ and set $\tilde{Q} \gets \gcd(Q^{(1)}, Q^{(2)}, \ldots)$. \;
Set $Q' \gets Q / \tilde{Q}$ and initialize $S \gets \{ Q'(X,Y) \}$. \;
\If{$Q'$ is \emph{useless} with respect to $\{\alpha_1, \dots, \alpha_n\}$ as per \autoref{def:useless}}      {\Return{$\emptyset$} }

\BlankLine
\nonl \textbf{Step 2: Refinement via \textsc{Split}.} \;
\While{there is a $g \in S$ of $Y$-degree at least $2$ that is not stable}{
  Choose a $j \in [n]$ such that $g$ is not stable with respect to $(\alpha_j,\beta_j)$\;
  Update $S \gets S \setminus \{ g \}$. \;
  Set $(g_1,g_2) \gets \Split(g,(\alpha_j,\beta_j))$\;
  \If{$\deg_Y(g_1) \ge 1$}{
    Write $g_1(X,Y) = \sum_{\ell} g_1^{(\ell)}(X) Y^\ell$ and set $\tilde{g}_1 \gets \gcd(g_1^{(1)}, g_1^{(2)}, \ldots)$. \;
    \If{$g_1 / \tilde{g}_1$ is not \emph{useless} with respect to $\{\alpha_1, \dots, \alpha_n\}$ as per \autoref{def:useless}}      {$S \gets S \cup \{ g_1 / \tilde{g}_1 \}$ }
    }
  \If{$\deg_Y(g_2) \ge 1$}{
    Write $g_2(X,Y) = \sum_{\ell} g_2^{(\ell)}(X) Y^\ell$ and set $\tilde{g}_2 \gets \gcd(g_2^{(1)}, g_2^{(2)}, \ldots)$. \;
   \If{$g_2 / \tilde{g}_2$ is not \emph{useless} with respect to $\{\alpha_1, \dots, \alpha_n\}$ as per \autoref{def:useless}}      {$S \gets S \cup \{ g_2 / \tilde{g}_2 \}$ }
    }
}

\BlankLine
\nonl \textbf{Step 3: Extract codewords.} \;
Initialize $L \gets \emptyset$. \;
\For{$g \in S$}{
  \If{$\deg(g) = 1$ and $g$ is of the form $\gamma(Y-f)$ for $\gamma \cdot \in \F\setminus \{0\}$, $deg(f) < k$} {
   \nonl \lIf{$f$ agrees with $w$ on at least $\sqrt{(k-1)n}$ coordinates}{
      \nonl $L \gets L \cup \{ f \}$.
  }}
  \Else{
  Set $U \gets \{ j \in [n] : g(\alpha_j,\beta_j)=0 \}$. \;
  \If{there exists $f \in \F[X]$ with $\deg(f) \le (k-1)$ and $f(\alpha_j)=\beta_j$ for all $j \in U$}{
    \nonl \lIf{$f$ agrees with $w$ on at least $\sqrt{(k-1)n}$ coordinates}{
      \nonl $L \gets L \cup \{ f \}$. 
    }
  }
  }
}
\Return{$L$}
\end{algorithm}

\subsection{Correctness of the algorithm}

We now argue the correctness of \cref{algo:det-gs-algo}.  

If the polynomial $Q'$ is useless with respect to $\{\alpha_1, \ldots, \alpha_n\}$, then, from \autoref{obs:uselss}, we get that $Q'$ has no factors of the form $Y-f$ for an $f \in \F[X]$ and the algorithm is clearly correct and terminates in polynomial time. Thus, for the rest of the this discussion, we assume that this is not the case. 

Suppose the while loop terminates after $T$ iterations, and let $S_t$ denote the set of polynomials at the end of the $t$-th iteration, for $t\in\{0,\dots,T\}$.  
Thus $S_0=\{Q'(X,Y)\}$ and $S_T$ is the final stable set.

To prove correctness, we proceed in two steps:
\begin{enumerate}
  \item For every $f$ in the true list $L$ and every iteration $t$, there exists some $g\in S_t$ such that $(Y-f(X))\mid g(X,Y)$.  
  \item For the final stable set $S_T$, the zeros of such a $g$ correspond exactly to the agreement points of $f$ with the received word $w$.  
\end{enumerate}
Since the number of agreement points exceeds $\sqrt{(k-1)n}>k-1$, the polynomial $f$ can be uniquely recovered by interpolation in Step~3 of \cref{algo:det-gs-algo}.  

We now formalize these claims, but before that state the following simple observation that is immediately evident from the description of the algorithm. 

\begin{observation}\label{obs:useful-and-content}
    For every $t \in \{1, 2, \ldots, T\}$, and for every polynomial $g(X,Y)$ in the set $S_t$, the following are true. 
    \begin{itemize}
        \item $g$ is not useless with respect to $\{\alpha_1, \ldots, \alpha_n\}$. 
        \item If $g = \sum_{\ell} g_i(X) Y^{\ell}$, then $\gcd(g_0(X), g_1(X), \ldots )$ equals $1$. In particular, $g$ has no pure $X$ factors. 
    \end{itemize}
\end{observation}

We now state and prove the first technical lemma. 
\begin{lemma}\label{lem:list-survive}
For every $f \in L$ and every $t \in \{0,\dots,T\}$, there exists a polynomial $g\in S_t$ such that $(Y-f(X))\mid g(X,Y)$.
\end{lemma}

\begin{proof}
We proceed by induction on $t$.  
For the base case $t=0$, the claim follows directly from \cref{lem:close-enough-polynomials-are-roots-GS}, since $(Y-f(X))\mid Q(X,Y)$ (and hence, $(Y-f)\mid Q'$) for every $f\in L$.

Assume the claim holds for $S_{t-1}$, and let $g\in S_{t-1}$ be such that $(Y-f(X))\mid g(X,Y)$.  
Consider the $t$-th iteration of the while loop, when $g$ is selected for splitting at some point $(\alpha_j,\beta_j)$. For this to happen the $Y$-degree of $g$ is at least two.  

If $g$ is already stable, it remains in $S_t$, and the inductive claim holds trivially.  
Otherwise, $(g_1,g_2)\gets\textsc{Split}(g,(\alpha_j,\beta_j))$.  
By \cref{obs:list-membership-category}, if $f(\alpha_j)=\beta_j$, then $(Y-f(X))$ divides $g_{A,j}$ ; otherwise it divides $g_{B,j}$. From \autoref{thm:split}, we get that in the first case, $(Y-f(X))$ divides $g_1$ and in the latter case, it divides $g_2$. 

If $Y-f$ divides $g_1$, then $g_1$ clearly has $Y$-degree at least one, and from \autoref{obs:uselss}, it cannot be useless with respect to $\{\alpha_1, \ldots, \alpha_n\}$, and it is added to the set $S_t$ (after removing some pure $X$ factors). Similarly, if $Y-f$ divides $g_2$, then $g_2$ clearly has $Y$-degree at least one, and from \autoref{obs:uselss}, it cannot be useless with respect to $\{\alpha_1, \ldots, \alpha_n\}$, and it is added to the set $S_t$ (again after removing some pure $X$ factors). 

Thus $S_t$ continues to contain a polynomial divisible by $(Y-f(X))$, completing the induction.
\end{proof}

\begin{lemma}\label{lem:zeros-match-agreements}
For every $f\in L$, there exists $g\in S_T$ such that, for all $j\in[n]$,
\[
  g(\alpha_j,\beta_j)=0 \iff f(\alpha_j)=\beta_j.
\]
\end{lemma}

\begin{proof}
From \cref{lem:list-survive}, there exists $g\in S_T$ with $(Y-f(X))\mid g(X,Y)$.  We work with this $g$ for the rest of the proof. 

Clearly, $g(\alpha_j,\beta_j)=0$ whenever $f(\alpha_j)=\beta_j$. Suppose for contradiction that there exists $j$ such that $g(\alpha_j,\beta_j)=0$ but $f(\alpha_j)\ne\beta_j$. Then $g$ must have some irreducible factor $\hat g(X,Y)$ distinct from $(Y-f(X))$ with $\hat g(\alpha_j,\beta_j)=0$ and $Y$-degree at least 1 (the latter is true because we have shaved off all pure $X$ factors). Furthermore, $\hat{g}(\alpha_j, Y)$ also has $Y$-degree at least one, since if this is not the case, the only way $\hat{g}(\alpha_j, \beta_j)$ can be zero is if $\hat{g}(\alpha_j, Y)$ is zero. But this would mean that $g$ has pure $X$-factors, hence contradicting \autoref{obs:useful-and-content}.  

Thus, the univariate $g(\alpha_j, Y)$ has $Y$-degree at least two, and is divisible by $(Y-f(\alpha_j))$ and $\hat{g}(\alpha_j, Y)$, both of which have $Y$ degree at least one. From \autoref{def:stability}, we get that $g$ is not stable at $(\alpha_j, \beta_j)$, contradicting the stability of $S_T$.  

Therefore, the zeros of $g(X,Y)$ coincide exactly with the agreement points of $f$ and $w$.
\end{proof}

By \cref{lem:zeros-match-agreements}, for every $f\in L$ there exists a polynomial $g\in S_T$ whose zero set identifies precisely the coordinates where $f$ and $w$ agree.  
Since the number of agreement points is greater than $\sqrt{(k-1)n}$ and thus exceeds $k-1$, $f$ can be uniquely recovered by interpolation in Step~3 of \cref{algo:det-gs-algo}.  

This completes the proof of correctness of the deterministic list-decoding algorithm.

\subsection{Running-time analysis}

The algorithm is clearly deterministic (assuming \cref{thm:split} for the \textsc{Split} subroutine).  
We now show that it runs in time $\poly(n, \log |\F|)$.  
The extraction step (Step~3) of \cref{algo:det-gs-algo} runs in $\poly(n, \log |\F|, |S|)$ time, where $S$ is the set of polynomials obtained after refinement (Step~2).  
Hence, it suffices to show that the refinement step (Step~2) terminates in polynomial time and that $|S|$ is polynomially bounded.

\begin{lemma}
The refinement step (Step~2) of \cref{algo:det-gs-algo} terminates in $\poly(\deg(Q), \log |\F|) = \poly(n, \log |\F|)$ time.
\end{lemma}
\begin{proof}
Let $S$ denote the set of polynomials at any stage of refinement step (Step~2).  
We now show, via a potential-function argument, that the refinement process terminates after only polynomially many updates.

Partition $S$ by $Y$-degree: let $S^{(i)} = \{g \in S \mid \deg_Y(g)=i\}$ for $i\in[\deg_Y(Q)]$.  
Define the potential
\[
  \Phi(S) = \sum_{i=1}^{\deg_Y(Q)} (i-1)\,|S^{(i)}|.
\]
Initially, $S=\{Q(X,Y)\}$ and $\Phi(S)=\deg_Y(Q)-1$.  
We will show that if $S$ is not already stable, then $\Phi(S)$ decreases by at least~$1$ in the next iteration, implying that the refinement loop executes at most $\deg_Y(Q)$ times.

Consider a refinement step where the set $S$ is modified to the set $S'$ by splitting $g\in S$ with respect to the point $(\alpha_j,\beta_j)$. For $g$ to be split, the $Y$-degree of $g$ must be at least two and there must be a point $(\alpha_j, \beta_j)$ such that $g$ is not stable at $(\alpha_j, \beta_j)$. Let $(g_1, g_2)$ be the output of \Split\ on input $g, (\alpha_j, \beta_j)$. 

Note that if neither of the polynomials $g_1' = g_1/\tilde{g_1}$ and  $g_2' = g_2/\tilde{g_2}$ is added to the set $S$ (to obtain $S'$) in this iteration, then the potential for $S'$ is clearly strictly less than that of $S$, since the polynomial $g$ with $Y$-degree at least two is removed from $S$ in the iteration. Thus, at least one of $g_1', g_2'$ is in $S'$. We now argue that the potential of $S'$ is strictly smaller than the potential of $S$ using some case analysis. 

Let $d=\deg_Y(g)$, $d_1 = \deg_Y(g_1')$ and $d_2 = \deg_Y(g_2')$. 
We now consider some cases. 

\begin{itemize}
    \item \textbf{$d_1 \geq 1, d_2 \geq 1$: }In this case, $\Phi(S') - \Phi(S) \leq -(d-1) + (d_1 -1 ) + (d_2-1)$, which is clearly at most $-1$ since $d_1 + d_2 \leq d$ from \autoref{cor:split-on-unstable-inputs}. 
    \item \textbf{$d_1 = d_2 = 0$: } In this case, $\Phi(S') - \Phi(S) = -(d-1)$ which is at most $-1$ since $d$ is at least $2$ (else $g$ could not have been chosen  for splitting).
    \item \textbf{$d_1 \geq 1, d_2 = 0$: } In this case, $\Phi(S') - \Phi(S) \leq -(d-1) + (d_1-1) = d_1 - d$, but this is again at most $-1$ since $d_1 < d$ from \autoref{cor:split-on-unstable-inputs}. 
    \item \textbf{$d_2 \geq 1, d_1 = 0$: } This is the same as the previous case. 
\end{itemize}

Since the initial potential is at most $\deg_Y(Q)-1$, the refinement loop can execute at most $\deg_Y(Q)$ iterations, each of which runs in $\poly(\deg(Q), \log|\F|)$ time.  
Hence, Step~2 terminates in $\poly(n,\log|\F|)$ time.
\end{proof}
  
\section{The \textsc{Split} algorithm}\label{sec:split}

In this section, we prove \cref{thm:split}, completing the proof of \cref{thm:det-gs} and hence of \cref{thm:intro-main}.  
For completeness, we restate the theorem below.

\splitthem*

\subsection{The Hensel lifting paradigm}\label{sec:hensel}

Our main ingredient is the classical technique of \emph{Hensel lifting}, which allows one to lift a local factorization of a polynomial modulo a power of $(X-\alpha)$ to a higher power.  
This may be viewed as the analog of Newton iteration for factoring, as opposed to root lifting.  
We recall the version we use below and refer to~\cite{SinhababuT2021,KoppartySS2015} for a proof.

\begin{lemma}[Hensel Lifting~{\cite{SinhababuT2021}}]\label{lem:HL}
Let $\F$ be a finite field, and let $P,g,h,a,b\in\F[X,Y]$, $\alpha\in\F$, and $m\ge1$ satisfy
\[
P \equiv g h \bmod{\langle X-\alpha\rangle^m}
\quad\text{and}\quad
a g + b h \equiv 1 \bmod{\langle X-\alpha\rangle^m}.
\]
Then the following hold.
\begin{enumerate}
  \item \text{(Existence)}  
  There exist $g',h'\in\F[X,Y]$ such that
  \begin{enumerate}
    \item $P\equiv g'h'\bmod{\langle X-\alpha \rangle^{2m}}$
    \item $g'\equiv g\bmod{\langle X-\alpha \rangle^m}$ and $h'\equiv h\bmod{\langle X-\alpha \rangle^m}$
    \item there exist $a',b'\in\F[X,Y]$ with $a'g' + b'h' \equiv 1 \bmod{\langle X-\alpha\rangle^{2m}}$. 
    \end{enumerate}
  We call $g',h'$ the \emph{lifts} of $g,h$ respectively and $a',b'$ the corresponding \emph{Bezout pairs}.
  \item \text{(Uniqueness)}  
  For any other lifts $g^*,h^*$ of $g,h$, there exists $u\in\langle X-\alpha\rangle^m$ such that
  \[
    g^* \equiv g'(1+u) \bmod{\langle X-\alpha\rangle^{2m}},
    \qquad
    h^* \equiv h'(1-u) \bmod{\langle X-\alpha\rangle^{2m}}.
  \]
  \item \text{(Efficiency)}  
  There is a deterministic algorithm that, given $P,g,h,a,b,m,\alpha$, outputs $g',h',a',b'$ in time $\poly(\deg(P),\deg(a),\deg(b),\log|\F|,m)$.
\end{enumerate}
\end{lemma}

\paragraph*{Degree bounds under iterative Hensel lifting.}
In standard applications of Hensel lifting for polynomial factorization, the polynomial $P$ is typically assumed to be monic, in which case degree growth during lifting is straightforward to bound.  
In our setting, we must handle the non-monic case as well.  
The following lemma provides quantitative degree bounds for the iterative lifting process; these follow from the analysis in~\cite[Lecture~10]{Sudan2012-AC}.

\begin{lemma}[Degree bounds for iterative Hensel lifting]\label{lem:hensel-degree-bounds}
Let $\F$ be any field, and let $P(X,Y)$ be a polynomial of total degree $d$.  
Fix $\alpha\in\F$.
\begin{enumerate}
  \item \textbf{(Degree bound)}  
  Let $g_0,h_0$ be relatively prime polynomials of positive $Y$-degree such that
  \[
    P \equiv g_0 h_0 \bmod{\langle X-\alpha\rangle}.
  \]
  Let $(g_i,h_i)$ and $(a_i,b_i)$ be the sequences of lifts and Bézout pairs obtained by successive applications of \cref{lem:HL}, satisfying
  \[
  \begin{aligned}
  &P \equiv g_i h_i \bmod{\langle X-\alpha\rangle^{2^{i}}}, \\
  &a_i g_i + b_i h_i \equiv 1 \bmod{\langle X-\alpha\rangle^{2^{i}}},\\
  &g_i \equiv g_{i-1} \bmod{\langle X-\alpha\rangle^{2^{i-1}}}, \quad
   h_i \equiv h_{i-1} \bmod{\langle X-\alpha\rangle^{2^{i-1}}}.
  \end{aligned}
  \]
  Then, for every $i$, the total degrees of $a_i,b_i,g_i,h_i$ are at most $d\cdot 5^i$.
  \item \textbf{(Uniqueness)}  
  Let $(g_i',h_i')$ be another sequence of pairs of total degree at most $D$ satisfying the same congruences as above and with $g_1'\equiv g_0$, $h_1'\equiv h_0$ modulo $\langle X-\alpha\rangle$.  
  Then there exist polynomials $u_0,u_1,\dots$ such that $u_i \in \langle X-\alpha\rangle^{2^i}$ and has total degree at most $(d+D)10^i$ satisfying
  \[
  g_i' \equiv g_i\cdot \left(\prod_{j<i}(1+u_j)\right)
  \bmod{\langle X-\alpha\rangle^{2^i}},
  \quad
  h_i' \equiv h_i\cdot \left(\prod_{j<i}(1+u_j)\right)^{-1}
  \bmod{\langle X-\alpha\rangle^{2^i}}.
  \]
\end{enumerate}
\end{lemma}

\begin{proof}
We closely follow the proof of \cref{lem:HL} from~\cite{SinhababuT2021}.
\paragraph*{Degree Bound:} From the proof of \cref{lem:HL} for any $i$ we have the following relations:
\begin{enumerate}
    \item $m_i = f - g_i \cdot h_i$;
    \item $g_{i+1} = g + b_i \cdot m_i$, \quad $h_{i+1} = g_i + a_i \cdot m_i$;
    \item $q_{i+1} = a_i \cdot g_{i+1} + b_i \cdot h_{i+1} - 1$;
    \item $a_{i+1} = a_i - a_i \cdot q_{i+1}$, \quad $b_{i+1} = b_i - b_i \cdot q_{i+1}$.
\end{enumerate}
For $i = 0$, we have $\deg(g_0), \deg(h_0) \le d$, since $\deg(P(\alpha,Y)) \le \deg(P) \le d$.  
As $g_0, h_0 \in \F[Y]$, the extended GCD algorithm ensures that $\deg(a_0) < \deg(h_0) \le d$ and $\deg(b_0) < \deg(g_0) \le d$. Hence, the base case holds.

Assume inductively that the total degrees of $(g_i, h_i, a_i, b_i)$ are at most $d \cdot 5^i$. Then $\deg(g_i\cdot h_i)\leq 2d\cdot 5^i$ and hence $\deg(m_i)\leq 2d\cdot 5^i$. Consequently $\deg(g_{i+1})\leq \deg(b_i\cdot m_i)\leq 3d\cdot 5^i$ and similarly, $\deg(h_{i+1})\leq 3d\cdot 5^i$. So, $\deg(a_i\cdot g_{i+1})\leq 4d\cdot 5^i$ and similarly, we have $\deg(b_i\cdot h_{i+1})\leq 4d\cdot 5^i$. Hence, $\deg(q_{i+1})\leq 4d\cdot 5^i$. Thus,  $\deg(a_{i+1})\leq \deg(a_i\cdot q_{i+1})\leq d\cdot 5^i+4d\cdot 5^i=d\cdot 5^{i+1}$ and similalry, $\deg(b_{i+1})\leq d\cdot 5^{i+1}$. So we got the degree bounds of $(g_{i+1},h_{i+1},a_{i+1},b_{i+1})$ are at most $d\cdot 5^{i+1}$. By induction, we conclude that for every $i \ge 0$, the total degrees of $g_i, h_i, a_i, b_i$ are bounded by $d \cdot 5^i$.

\paragraph*{Uniqueness:} We now establish uniqueness by induction on $i$. For the base case $i=1$, we have $g_1\equiv g_1'\equiv g_0\bmod{\langle X-\alpha\rangle}$ and $h_1\equiv h_1'\equiv h_0\bmod{\langle X-\alpha\rangle}$. Following the proof of \cref{lem:HL}  define $u_0=a_1(g_1'-g_1)-b_1(h_1'-h_1)$, so that  $u_0\in\bmod{\langle X-\alpha\rangle}$ and 
\[
g_1'\equiv g_1(1+u_0)\bmod{\langle X-\alpha\rangle^2}, \quad h_1'\equiv h_1(1-u_0)\equiv h_1(1+u_0)^{-1}\bmod{\langle X\rangle^2}
\]
From the degree bounds above, the total degrees of $(g_1,h_1,a_1,b_1)$ are at most $5d$, implying $\deg(a_1(g_1'-g_1)), \deg(b_1(g_1'-g_1))\leq (d+D)10$ and hence $\deg(u_0)\leq (d+D)10$. Hence, the base case is satisfied. 

Now suppose that this is true for $i$. Since $u_j \in \langle X - \alpha \rangle^{2^j}$ for each $j < i$, it follows that $(1 + u_j)$ is invertible modulo $\langle X - \alpha \rangle^{2^i}$. Now consider the elements $g_{i+1}'\left(\prod_{j<i}(1+u_j)\right)^{-1}$ and $h_{i+1}'\left(\prod_{j<i}(1+u_j)\right)$. We have $g_{i+1}'\left(\prod_{j<i}(1+u_j)\right)^{-1}\equiv g_i'\left(\prod_{j<i}(1+u_j)\right)^{-1}\equiv g_i\bmod{\langle X-\alpha\rangle^{2^i}}$ and $h_{i+1}'\left(\prod_{j<i}(1+u_j)\right)\equiv h_i'\left(\prod_{j<i}(1+u_j)\right)\equiv h_i\bmod{\langle X-\alpha\rangle^{2^i}}$. We also have 
\[
g_{i+1}'\left(\prod_{j<i}(1+u_j)\right)^{-1}\cdot h_{i+1}'\left(\prod_{j<i}(1+u_j)\right)\equiv g_{i+1}'h_{i+1}'\equiv f\bmod{\langle X-\alpha\rangle^{2^{i+1}}}
\]
Thus, both $g_{i+1},h_{i+1}$ and $g_{i+1}'\left(\prod_{j<i}(1+u_j)\right)^{-1},h_{i+1}'\left(\prod_{j<i}(1+u_j)\right)$ are valid lifts of $(g_i,h_i)$. By the Uniqueness property of \cref{lem:HL}, there exists $u_i\in \langle X-\alpha\rangle^{2^i}$ such that 
\[
g_{i+1}'\left(\prod_{j<i}(1+u_j)\right)^{-1}\equiv g_{i+1}(1+u_i)\bmod{\langle X-\alpha\rangle^{2^{i+1}}}
\]
and therefore $ g_{i+1}'\equiv g_{i+1}\prod_{j<i+1}(1+u_j)\bmod{\langle X-\alpha\rangle^{2^{i+1}}}$. Similarly, 
\[
h_{i+1}'\left(\prod_{j<i}(1+u_j)\right)\equiv h_{i+1}(1-u_i)\bmod{\langle X-\alpha\rangle^{2^{i+1}}}
\]
which implies $ h_{i+1}'\equiv h_{i+1}\left(\prod_{j<i+1}(1+u_j)\right)^{-1}\bmod{\langle X-\alpha\rangle^{2^{i+1}}} $. Since $(1+u_i)(1-u_i)\equiv 1\bmod{\langle X-\alpha\rangle^{2^{i+1}}}$. From the construction in  \cref{lem:HL}, we have $u_i=a_{i+1}(g_{i+1}'-g_{i+1})-b_{i+1}(h'_{i+1}-h_{i+1})$ and, by the degree bound in the first part, we have the total degrees of $(g_{i+1},h_{i+1}, a_{i+1},b_{i+1})$ are at most $d\cdot 5^{i+1}$. Therefore, $\deg(u_i)\leq d\cdot 5^{i+1}+\max\{d\cdot 5^{i+1},D\}\leq (d+D)10^{i+1}$. Hence, by induction, for all $i$, $g_{i}' \equiv g_{i} \left(\prod_{j < i}(1+u_j)\right) \bmod{\langle X-\alpha \rangle}^{2^{i}}$ and $h_{i}' \equiv h_{i} \cdot \left(\prod_{j < i}(1-u_j)\right)^{-1} \bmod{\langle X-\alpha \rangle}^{2^{i}}$ with $\deg(u_i) \le (d + D) \cdot 10^i$, as desired.
 \end{proof}

\subsection{Details of the \textsc{Split} algorithm}
We now describe the algorithm promised in \cref{thm:split}. For clarity, from now on we write $P_A,P_B,P_C,P_D$ for $P_{A(\alpha,\beta)},P_{B(\alpha,\beta)},P_{C(\alpha,\beta)}, P_{D(\alpha,\beta)}$, respectively.

Before proceeding, we record a basic but useful subroutine that we will invoke repeatedly. Suppose the polynomial $P \in \F[X,Y]$ has a nontrivial factor $g \in \F[X,Y]$. If we already know the splits of $g$ and of $P/g$ with respect to a point $(\alpha,\beta)$, then we can obtain the split of $P$ itself by simply combining the corresponding $A$ and $B$ parts. The following routine formalizes this observation.

\begin{algorithm}\label{algo:Combine-Splits}
\caption{\textsc{Combine-Splits}}
\DontPrintSemicolon
\SetAlgoNoEnd
\KwIn{Polynomials $P,g \in \F[X,Y]$ and $(\alpha,\beta)\in\F^2$ with the guarantee that $g$ divides $P$}
\KwOut{Polynomials $P_{A(\alpha,\beta)}\cdot R_1$ and $ P_{B(\alpha,\beta)}\cdot R_2$ for some polynomials $R_1,R_2$ such that $R_1\cdot R_2$ divide $P_{D(\alpha,\beta)}$}
\nonl\hrulefill\\
\BlankLine
Set $g' \gets P/g$\;
Compute $(g_1,g_2) \gets \textsc{Split}(g,(\alpha,\beta))$ and $(g'_1,g'_2) \gets \textsc{Split}(g',(\alpha,\beta))$\;
\Return $(g_1\cdot g'_1,\; g_2\cdot g'_2)$\;
\end{algorithm}

We now develop a series of claims that naturally lead to the \Split\ algorithm. 
We begin with a deterministic base case that seeds the Hensel lifting.

\begin{lemma}\label{lem:base-case-split}
Let $\F$ be any finite field. There is a deterministic algorithm that, given a bivariate polynomial $P(X,Y)\in\F[X,Y]$ and a point $(\alpha,\beta)\in\F^2$, runs in time $\poly(\deg(P),\log|\F|)$ and outputs a factorization
\[
  P(\alpha,Y) = (Y-\beta)^m \cdot \widehat{P}(Y),
\]
where $\widehat{P}(\beta)\ne 0$.
\end{lemma}

\begin{proof}
Iteratively perform univariate division to find the largest $m\in\mathbb{N}$ such that $(Y-\beta)^m$ divides $P(\alpha,Y)$; then set $\widehat{P}(Y):=P(\alpha,Y)/(Y-\beta)^m$. Both steps are deterministic and run in time $\poly(\deg(P),\log|\F|)$ by standard univariate algorithms.
\end{proof}

An immediate consequence is that if $\deg_Y(P(\alpha,Y))\geq 1$, we can recognize when $P=P_{A}P_{D}$ or $P=P_{B}P_{D}$: if $m=\deg_Y(P)$, then $P=P_{A}P_{D}$; if $m=0$, then $P=P_{B}P_{D}$. In either case, \cref{thm:split} follows directly. If $\deg_Y(P(\alpha,Y))=0$ and $P(\alpha,\beta)\neq 0$, $P=P_D$. In that case we just throw it away.

The interesting case is when $0<m<\deg_Y(P)$. Then \cref{lem:base-case-split} yields
\[
  P(\alpha,Y) = (Y-\beta)^m \cdot \widehat{P}(Y),
\]
where $(Y-\beta)^m$ and $\widehat{P}(Y)$ are coprime and both have positive $Y$-degree. Equivalently,
\[
  P(X,Y) \equiv (Y-\beta)^m \cdot \widehat{P}(Y) \pmod{\langle X-\alpha\rangle}.
\]
This is precisely the hypothesis needed to invoke \cref{lem:HL}: we lift this factorization to sufficiently high precision in $(X-\alpha)$, recover a nontrivial factor of $P$, and then recurse, using \CombSplits\ to combine results.

More precisely, set $g_0=(Y-\beta)^m$ and $h_0=\widehat{P}(Y)$. Apply iterative Hensel lifting (\cref{lem:HL}) for
$t :=2\left\lceil \log \deg_Y(P) \right\rceil + 1$
rounds to obtain $g_t,h_t$ with
\[
  P(X,Y) \equiv g_t\cdot h_t \bmod{\langle X-\alpha\rangle^{2^t}},\qquad
  g_t \equiv g_0 \bmod{\langle X-\alpha\rangle},\quad h_t \equiv h_0 \bmod{\langle X-\alpha\rangle}.
\]

Given these lifts $g_t,h_t$, we can write the following linear system
\begin{align*}
U\cdot h_t &\equiv V \bmod{\langle X-\alpha\rangle^{2^t}} ,\\
E\cdot g_t &\equiv F \bmod{\langle X-\alpha\rangle^{2^t}} ,
\end{align*}
where the variables are the coefficients of the polynomials $U,V,E,F$. The following claim is standard in Hensel-based factorization and we include it for completeness.

\begin{claim}\label{clm:nontrivial-gcd}
If the system $U\cdot h_t \equiv V \bmod{\langle X-\alpha\rangle^{2^t}}$ has a solution, then $\gcd(V,P)$ has $Y$-degree at least one. Similarly, if the system $E\cdot g_t \equiv F \bmod{\langle X-\alpha\rangle^{2^t}}$ has a solution, then $\gcd(F,P)$ has $Y$-degree at least one.
\end{claim}
Finally, we show that if both systems have no solution, then $P$ must be entirely in the $C$ $D$-parts.

\begin{claim}\label{clm:linear-system-no-sol}
Let $P(X,Y)$ and $(\alpha,\beta)$ be inputs with $P \ne P_{A}P_{D}$ and $P \ne P_{B}P_{D}$. If neither
\[
  U\cdot h_t \equiv V \bmod{\langle X-\alpha\rangle^{2^t}}
  \quad\text{nor}\quad
  E\cdot g_t \equiv F \bmod{\langle X-\alpha\rangle^{2^t}}
\]
has a solution, then $P = P_{C}P_{D}$.
\end{claim}

Given all the above observations and claims,  we now have the following deterministic \textsc{Split} algorithm.

\begin{algorithm}[H]
\caption{\textsc{Split}}\label{algo:split}
\DontPrintSemicolon
\SetAlgoNoEnd
\KwIn{$P(X,Y)\in\F[X,Y]$ and a point $(\alpha,\beta)\in\F^2$}
\KwOut{Polynomials $P_{A(\alpha,\beta)}\cdot R_1$ and $P_{B(\alpha,\beta)}\cdot R_2$, for some polynomials $R_1,R_2$ such that $R_1\cdot R_2$ divide $P_{D(\alpha,\beta)}$.}
\nonl\hrulefill
\BlankLine

\nonl \textbf{Step 0: Remove pure-$X$ factors.}\;
Write $P(X,Y)=\sum_{\ell} P^{(\ell)}(X)\,Y^\ell$ and set $P \gets P / \gcd(P^{(1)},P^{(2)},\ldots)$.\;
\BlankLine

\nonl \textbf{Step 1: Seed Hensel lifting and handle simple cases.}\;
Factor $P(\alpha,Y)=(Y-\beta)^m \cdot \widehat{P}(Y)$ such that $\widehat{P}(\beta) \neq 0$ \Comment*[r]{Use \cref{lem:base-case-split}}
Set $g_0:=(Y-\beta)^m$, $h_0:=\widehat{P}(Y)$.\;
\lIf(\tcc*[f]{$P=P_D$}){$\deg(P(\alpha,Y))=0$}{\Return{$(1,1)$}}
\lIf(\tcc*[f]{{$P=P_AP_D$}}){$\deg(g_0)=\deg P(\alpha,Y)$}{\Return{$(P,1)$}}
\lIf(\tcc*[f]{{$P=P_BP_D$}}){$\deg(h_0)=\deg P(\alpha,Y)$}{\Return{$(1,P)$}}
\BlankLine

\nonl \textbf{Step 2: Iterative Hensel lifting.}\;
Set $t \gets 2\lceil\log(\deg_Y P)\rceil+1$ and $\widehat{D} \gets 2\deg(P)\cdot 10^{3\log \deg(P)}$.\;
Apply iterative Hensel lifting (\cref{lem:HL}) for $t$ rounds to compute lifts $g_t,h_t$ satisfying\;
\nonl \Indp $P \equiv g_t \cdot h_t \bmod{\langle X-\alpha\rangle^{2^t}}, \qquad g_t \equiv g_0 \bmod{\langle X-\alpha\rangle},
\qquad h_t \equiv h_0 \bmod{\langle X-\alpha\rangle}$.\;
\BlankLine

\nonl \Indm \textbf{Step 3: Solve linear systems.}\;
\If{there exist $U,V\in\F[X,Y]$ satisfying $U \cdot h_t \equiv V \bmod{\langle X-\alpha\rangle^{2^t}}$ with the degree constraints $\deg_Y U \le \widehat{D}$, $\deg_Y V \le \deg_Y P - 1$, $\deg_X U \le 2^t-1$, $\deg_X V \le \deg_X P$}{
  Set $\rho \gets \gcd(P,V)$ \Comment*[r]{$\deg_Y \rho \ge 1$ by \cref{clm:nontrivial-gcd}}
  \Return{\CombSplits$(P,(\alpha,\beta),\rho)$}
}
\If{there exist $E,F\in\F[X,Y]$ satisfying $E\cdot g_t \equiv F \bmod{\langle X-\alpha\rangle^{2^t}}$ with the degree constraints $\deg_Y E \le \widehat{D}$, $\deg_Y F \le \deg_Y P - 1$, $\deg_X E \le 2^t-1$, $\deg_X F \le \deg_X P$ such that }{
  Set $\sigma \gets \gcd(P,F)$ \Comment*[r]{$\deg_Y \sigma \ge 1$ by \cref{clm:nontrivial-gcd}}
  \Return{\CombSplits$(P,(\alpha,\beta),\sigma)$}
}
\BlankLine

\nonl \textbf{Step 4: All tests fail.} \Comment*[r]{Use \cref{clm:linear-system-no-sol}}
\Return{$(1,1)$}.
\end{algorithm}

\subsection{Correctness of the \Split\ algorithm}
We now establish the correctness of the \Split\ algorithm. We first prove the main correctness lemma assuming \cref{clm:nontrivial-gcd,clm:linear-system-no-sol}, and subsequently provide proofs of these claims.

\begin{lemma}[Correctness of \Split\ algorithm]
For every bivariate polynomial $P(X,Y)$ and point $(\alpha,\beta)\in\F^2$, \cref{algo:split} terminates and correctly outputs the decomposition $P_{A}R_1$ and $P_{B}R_2$ where $R_1\cdot R_2$ divides $P_{D}$.
\end{lemma}

\begin{proof}

We proceed by induction on $\deg_Y P$.

\textit{Base case.} If $\deg_Y P = 1$, then $P$ must equal one of $P_A,P_B, P_D$, and the algorithm’s initial checks clearly return the correct classification.

\textit{Inductive step.} Assume the algorithm is correct for all polynomials of $Y$-degree at most $d-1$. Let $P$ have $\deg_Y P = d>1$.

If $P=P_AP_D$ or $P=P_BP_D$, correctness is immediate. Otherwise, both $g_0$ and $h_0$ from Step~1 have positive $Y$-degree, and the algorithm proceeds to the Hensel lifting step to obtain $g_t,h_t$ and then further proceeds to solving the linear systems. 

If either linear system admits a solution, then by \cref{clm:nontrivial-gcd}, the resulting GCD has positive $Y$-degree, producing a smaller subproblem on which the inductive hypothesis applies.  If neither system has a solution, then by \cref{clm:linear-system-no-sol}, $P=P_CP_D$, and the algorithm correctly returns $(1,1)$.

We now prove the correctness of the algorithm if $P$ is \emph{stable}. Since $P$ is stable, $\deg(P(\alpha,Y))\geq 1$. Suppose after invoking \cref{lem:base-case-split} we obtain the factorization $P(\alpha,Y)=(Y-\beta)^m\cdot \hat{P}(Y)$ such that $\hat{P}(\beta)\neq 0$. Now if $P(\alpha,Y)=\gamma (Y-\beta)^r$ for some $\gamma\in\F\setminus\{0\}$ and positive integer $r$, then, $r=m$ and $\hat{P}(Y)=\gamma$. Therefore, $\deg(P(\alpha,Y))=m$ hence the algorithm returns $(P,1)$ in Step~1. If $P(\alpha,\beta)\neq 0$, $m=0$. Therefore $\deg(P(\alpha,Y))=\deg(\hat{P})$. Therefore, the algorithm returns $(1,P)$ in Step~1.

This completes the inductive argument, and hence the proof of correctness.
\end{proof}

We now prove \cref{clm:nontrivial-gcd,clm:linear-system-no-sol}.

\begin{proof}[Proof of \cref{clm:nontrivial-gcd}]
Consider the resultant $\Gamma(X):=\Res_Y(P,V)$. We have $\deg \Gamma \le 2\deg_X(P)^2$. Since $t=2\lceil\log(\deg_Y P)\rceil+1$, we have $2^t > \deg \Gamma$. From the property of resultants, there exist $a,b\in\F[X,Y]$ with
\[
  a\cdot P + b\cdot V = \Gamma(X).
\]
Reducing modulo $\langle X-\alpha\rangle^{2^t}$ and using $P \equiv g_t \cdot h_t \bmod \langle X-\alpha\rangle^{2^t}$ and $V \equiv U \cdot h_t \bmod \langle X-\alpha\rangle^{2^t}$ yields
\[
  \left(a \cdot g_t + b \cdot U\right) \cdot h_t \equiv \Gamma \bmod{\langle X-\alpha\rangle^{2^t}}.
\]
Viewing both sides as polynomials in $(X-\alpha)$ with coefficients in $\F[Y]$, the constant term of $h_t$ equals $h_0(Y)$, which has $\deg_Y\ge 1$. However, the right hand side, namely $\Gamma = \Res_Y(P,V)$ has no dependence on $Y$. Hence, it must be the case that $a g_t + b U$ is divisible by $(X-\alpha)^{2^t}$, forcing $(X-\alpha)^{2^t}\mid \Gamma$. But $\deg \Gamma < 2^t$, so this implies $\Gamma\equiv 0$, i.e., $P$ and $V$ have a nontrivial common factor in $\F(X)[Y]$ (By \cref{lem:res}). By Gauss’s lemma (see \cref{lem:gauss}), they have a nontrivial GCD in $\F[X,Y]$ of positive $Y$-degree. The case of $E,g_t,F$ is analogous.
\end{proof}

\begin{proof}[Proof of \cref{clm:linear-system-no-sol}]
For brevity, assume wlog $\alpha=0$ (so the ideal is $\langle X\rangle$). By assumption, $P\ne P_AP_D$ and $P\ne P_BP_D$. We prove the contrapositive: assume $P\ne P_CP_D$ and show that at least one of the two linear systems has a solution.

Since $P=P_A P_B P_CP_D$ and $P\ne P_CP_D$, at least two of $P_A,P_B,P_C$ have positive $Y$-degree.

\emph{Case I:} $\deg_Y P_A \ge 1$ and $\deg_Y(P_B P_C)\ge 1$. Reducing modulo $\langle X\rangle$,
\[
  P_B(X,Y) P_C(X,Y)P_D(X,Y) \equiv (Y-\beta)^m \hat h_0(Y) \pmod{\langle X\rangle}
\]
with $\hat h_0(\beta)\ne 0$ and $m\ge 1$. Here $\hat h_0 = h_0$ (since $P_A$ contributes only powers of $(Y-\beta)$). Let $\hat g_0=(Y-\beta)^m$. Hensel lifting gives lifts $\hat g_t,\hat h_t$ with
\[
  P_B P_CP_D \equiv \hat g_t \hat h_t \pmod{\langle X\rangle^{2^t}},\qquad
  \hat g_t \equiv \hat g_0,\ \hat h_t \equiv \hat h_0 \pmod{\langle X\rangle}.
\]
Hence
\[
  P \equiv \big(P_A \hat g_t\big)\,\hat h_t \pmod{\langle X\rangle^{2^t}}.
\]
Both this factorization and $P \equiv g_t h_t \pmod{\langle X\rangle^{2^t}}$ lift $P(0,Y)=g_0(Y)h_0(Y)$. Since $\hat h_t \equiv h_0 \pmod{\langle X\rangle}$, comparison modulo $\langle X\rangle$ yields $P_A(0,Y)\,\hat g_0(Y) = g_0(Y)$. Using \cref{lem:hensel-degree-bounds} (degree and uniqueness parts), there exist $u_1,\dots,u_t$ with controlled total degrees and vanishing modulo the appropriate powers of $X$ such that
\[
  \hat h_t \equiv h_t \prod_{j<t} (1+u_j) \pmod{\langle X\rangle^{2^t}}.
\]
Set
\[
  U := \hat g_t \left(\prod_{j<t} (1+u_j)\right) \bmod \langle X\rangle^{2^t},\qquad
  V := P_B P_C.
\]
Then $U\cdot h_t \equiv V \pmod{\langle X\rangle^{2^t}}$, with $\deg_X U \le 2^t-1$, $\deg_X V \le \deg_X P$, and (by the chosen parameter bounds) $\deg_Y U \le \widehat D$, $\deg_Y V < \deg_Y P$. Thus the first linear system has a valid solution.

\emph{Case II:} $\deg_Y P_B \ge 1$ and $\deg_Y(P_A P_C)\ge 1$. An entirely analogous argument (now producing a solution to $E\cdot g_t \equiv F \pmod{\langle X\rangle^{2^t}}$) yields a solution to the second system.

In either case, at least one system has a solution. Therefore, if neither system has a solution, necessarily $P=P_CP_D$.
\end{proof}

\subsection{Running time analysis}
The algorithm described above is clearly deterministic. We now analyze its running time.

The computations  in Step~0 and Step~1 clearly take at most $\poly(\deg(P),\log|\F|)$ time.  
From \cref{lem:HL}, each Hensel lifting step, that is, computing the lift from $(g_{i-1},h_{i-1})$ to $(g_i,h_i)$, requires $\poly(\deg(P),\log|\F|)$ time and this is performed for at most $t = \Theta(\log \deg_Y P)$ rounds.  
In Step~3, the systems of linear equations each have size polynomial in $\deg(P)$, and we solve $O(\deg(P))$ such systems.  
Hence this step also runs in $\poly(\deg(P),\log|\F|)$ time overall.

We now focus on the recursive structure of the algorithm.  
The key observation is that at most two recursive calls are made before termination.  
Each call is made with $P$ replaced by polynomials of the form $g$ and $P/g$, where $g$ is a nontrivial factor of $P$ whose $Y$-degree is at least one and strictly smaller than $\deg_Y(P)$.  
Assuming this, the total running time satisfies the recurrence
\[
  T(d) \;\le\; \max_{r\in\{1,\dots,d-1\}}\big\{\,T(r)+T(d-r)+\poly(d,\log|\F|)\,\big\}.
\]
This immediately implies $T(d)=\poly(d,\log|\F|)$.

To justify the degree constraint, observe first that it holds for the pairs $(\rho,P/\rho)$ and $(\sigma,P/\sigma)$ arising from Step~3, since $\deg_Y(\rho)\le\deg_Y(V)$ and $\deg_Y(V)<\deg_Y(P)$ by \cref{clm:nontrivial-gcd}, and likewise for $\sigma$.

Therefore, every recursive call reduces the $Y$-degree strictly, and each level of recursion incurs only polynomial work.  
This proves that the total running time of the algorithm is $\poly(\deg(P),\log|\F|)$.

  

\subsection*{Acknowledgements}
We thank Daniel Augot and Romain Lebreton for pointing out the reference \cite{AugotP2000}, which proves a derandomization of Sudan's algorithm for Reed-Solomon code (and their algebraic geometric variants). We were unaware of this reference at the time of the initial posting of our work.

Mrinal thanks Swastik Kopparty for many helpful discussions on the questions and themes studied in this paper and much encouragement.

\addcontentsline{toc}{section}{References}
{\small
  \bibliographystyle{prahladhurl}
  \bibliography{deterministicGS-bib}
}

\end{document}